\documentclass[conference]{IEEEtran}

\usepackage[utf8]{inputenc}
\usepackage{verbatim}
\usepackage{fancyvrb}
\PassOptionsToPackage{hyphens}{url}\usepackage{hyperref}
\usepackage{xcolor}
\usepackage{todonotes}
\usepackage{amsmath}
\usepackage{amssymb}
\usepackage{amsthm}
\usepackage{dirtytalk}
\usepackage{tikz}

\def\cert[#1](#2#3#4#5){C^{#1}_{#2\!#3,#4\!#5}}
\def\token[#1](#2#3){T^{#1}_{#2,#3}}

\newcommand{\myC}{${\mathcal C}$}
\newcommand{\myS}{${\mathcal S}$}
\newcommand{\myH}{${\mathcal H}$}
\newcommand{\myHp}{${\mathcal H}'$}
\newcommand{\myD}{${\mathcal D}$}
\newcommand{\myR}{${\mathcal R}$}
\newcommand{\myK}{${\mathcal K}$}
\newcommand{\myKp}{${\mathcal K}'$}
\newcommand{\myA}{${\mathcal A}$}
\newcommand{\myF}{${\mathcal F}$}
\newcommand{\myB}{${\mathcal B}$}

\newcommand{\myW}{${\mathcal W}$}
\newcommand{\mySc}{${\mathcal S}_C$}

\newcommand{\mycookie}{$\omega$}
\newcommand{\myno}{$\times$}

\newcommand{\tokS}{$\token[M]({\mathcal S} \mathcal S)$}

\newcommand{\iffdef}{\overset{\mathrm{def}}{\Leftrightarrow}}

\theoremstyle{definition}
\newtheorem{definition}{Definition}[section]
\newtheorem{theorem}{Theorem}[section]

\title{Analysis of the Security Design, Engineering, and\\ 
Implementation of the SecureDNA System} 

\author{\IEEEauthorblockN{Alan T. Sherman, Jeremy J. Romanik Romano,\\ 
Edward Zieglar, Enis Golaszewski, Jonathan D. Fuchs}
\IEEEauthorblockA{Cyber Defense Lab\\
University of Maryland, Baltimore County (UMBC)\\
Baltimore, Maryland 21250\\
Email:  \{sherman,jeremyr2,eziegl1,golaszewski,jfuchs2\}@umbc.edu }
\and
\IEEEauthorblockN{William E. Byrd}
\IEEEauthorblockA{Hugh Kaul Precision Medicine Institute\\
Heersink School of Medicine\\
University of Alabama at Birmingham\\
Birmingham, Alabama 35294\\  
Email:webyrd@uab.edu}
}

\date{December 9, 2025}  

\begin{document}

\IEEEoverridecommandlockouts
\makeatletter\def\@IEEEpubidpullup{6.5\baselineskip}\makeatother
\IEEEpubid{\parbox{\columnwidth}{
		Network and Distributed System Security (NDSS) Symposium 2026\\
		23-27 February 2026, San Diego, CA, USA\\
		ISBN XXX-X-XXXXXXX-X-X\\
		https://dx.doi.org/10.14722/ndss.2026.[23$|$24]xxxx\\
		www.ndss-symposium.org
}
\hspace{\columnsep}\makebox[\columnwidth]{}}

\maketitle

\begin{abstract}
We analyze security aspects of the SecureDNA system 
regarding its system design, engineering, and implementation.
This system enables DNA synthesizers to screen order requests against a database of hazards. 
By applying novel cryptography involving distributed oblivious pseudorandom functions, 
the system aims to keep order requests and the database of hazards secret.
Discerning the detailed operation of the system in part from source code (Version 1.0.8),
our analysis examines key management, certificate infrastructure, authentication, and rate-limiting mechanisms.
We also perform the first formal-methods analysis of the 
mutual authentication, basic request, and exemption-handling protocols.

Without breaking the cryptography, our main finding is that
SecureDNA's custom mutual authentication protocol SCEP achieves only one-way authentication: 
the hazards database and keyservers never learn with whom they communicate.
This structural weakness violates the principle of defense in depth and 
enables an adversary to circumvent rate limits that
protect the secrecy of the hazards database,
if the synthesizer connects with
a malicious or corrupted keyserver or hashed database.
We point out an additional structural weakness
that also violates the principle of defense in depth:
inadequate cryptographic bindings prevent the system from detecting if responses, 
within a TLS channel, from the hazards database were modified.
Consequently, if a synthesizer were to reconnect with the database
over the same TLS session, an adversary could replay and swap
responses from the database without breaking TLS.
Although the SecureDNA implementation does not allow such reconnections, it would be stronger
security engineering to avoid the underlying structural weakness.
We identify these vulnerabilities and
suggest and verify mitigations, including adding strong bindings.
Software Version 1.1.0 fixes SCEP with our proposed SCEP+ protocol.

Our work illustrates that a secure system needs more than sound mathematical cryptography; 
it also requires formal specifications, sound key management, 
proper binding of protocol message components, and
careful attention to engineering and implementation details.
\end{abstract}

\section{Introduction} 
\label{sec:intro}

The combination of gene editing technology (e.g., CRISPR~\cite{barrangou2016applications}),
DNA synthesis, and AI~\cite{AI2027}
poses an existential threat to humanity.
Advances now enable a single malicious actor, with modest resources, to synthesize pathogens capable of starting pandemics.
As an initial response to this threat, an international
group of eminent biologists and cryptographers 
designed and implemented the SecureDNA system~\cite{SecureDNA2024}, 
which provides a method for honest DNA synthesis labs 
to screen synthesis-order requests for known biohazards. 
In late 2022, the SecureDNA Foundation~\cite{SecureDNAweb} launched the system.

SecureDNA is significant for its advanced technology and impactful policy. It is the first
fully privacy-preserving, cryptographically verifiable screening system.
It operates in near real time through a distributed network of independent authorities
and supports standardized verifiable compliance. 
SecureDNA's long-term strategy is to push for regulation mandating universal screening of DNA synthesis orders.
Unfortunately, currently a malicious entity could avoid SecureDNA controls by
purchasing their own synthesis machine for less than \$50,000, or by
sending their request to a lab that does not screen.

Using a \textit{distributed oblivious pseudorandom function (DOPRF)}, the system
enables synthesizers to screen order requests against a curated database of hazards 
while keeping the database and requests secret~\cite{baum2024CIC,baum2020cryptographic}. 
To eliminate a possible single critical point of failure, the system distributes the key used by the DOPRF among $n$ keyservers using Shamir secret sharing~\cite{shamir1979share,stinson92}.

We analyze security aspects of the SecureDNA system, 
focusing on its system design, engineering, and implementation.
As part of this study, we perform the first formal-methods symbolic analysis 
of the two main protocols in SecureDNA for \textit{structural} (fundamental logical) weaknesses.  
These protocols are the \textit{basic order-request protocol} and the
\textit{exemption-handling protocol}, which permits a qualifying customer to order certain dangerous sequences.
Our work studies the security of the SecureDNA system design and implementation. 
We do not review the bio-design, high-level crypto-design, or low-level software security.
We do not examine the non-public monitoring and auditing that play
an important role in operational system security.

Because the available written descriptions do not adequately describe the protocols, 
we first discern these protocols in part by examining the source code (Version 1.0.8)~\cite{sec_dna_github}.
We also examine PCAP data from our working local copy of the system.
We then model these protocols and precisely state security goals.  
Using the \textit{Cryptographic Protocol Shapes Analyzer (CPSA)}~\cite{cpsamanual},
we analyze whether the models achieve the goals.
Our analysis includes a careful examination of 
SecureDNA's key and certificate management systems
and their use in authentication and rate-limiting
of queries to the hazards database.

Despite rigorous \textit{universal composability (UC) proofs}~\cite{canetti2020universally}
of the abstract cryptography~\cite{baum2024CIC}, 
the concrete protocols in use in the SecureDNA system
lack detailed descriptions and formal models, and 
the written descriptions lack verifiable, precisely-stated security goals of the system.
Thus, our starting point is largely the SecureDNA system's
whitepaper~\cite{SecureDNA2024} and publicly-released source code~\cite{sec_dna_github}.

Experience shows that there is great value in formal-methods analysis
of cryptographic protocols, including in the design phases, because
humans are poorly suited to analyze their complex nuanced
security properties and their many possible execution sequences.
For example, in 1995, using a protocol analysis tool, Lowe~\cite{Lowe1995} identified
a protocol-interaction attack on the 1978 \textit{Needham-Schroeder (NS)} public-key authentication protocol~\cite{Needham1978}, based on
a subtle structural flaw which had gone unnoticed for 17 years.
We point out a similar attack on SecureDNA's basic request protocol
that circumvents rate limiting.
We first discovered this flaw through our formal-methods analysis.

Aside from its monitoring, SecureDNA depends crucially on four 
mechanisms to achieve its security goals:
a \textit{custom mutual-authentication protocol (SCEP)}, 
\textit{exemption-list tokens (ELTs)} to implement its exemption-list feature,
an associated \textit{certificate management system}, and
{\it Transport Layer Security (TLS)}~\cite{rescorla2001ssl,rfc5246,rfc8446}
to protect communication channels. SecureDNA uses TLS 1.3 and 1.2, but only with ECDH and not with RSA.
If TLS fails, then a \textit{Man-in-the-middle (MitM)} attack would be possible and the only security goal that SecureDNA would achieve would be secrecy of orders.


Our analysis uncovers undesirable vulnerabilities concerning how SecureDNA uses these mechanisms, violating the principle of defense in depth.
First, due to a structural weakness, SecureDNA's custom authentication protocol provides only one-way authentication: the hazards database and keyservers never learn with whom they communicate.  
Consequently, an adversary, who tricks an honest synthesizer {\myS}
into connecting with a malicious or corrupted database or keyserver,
could masquerade as {\myS} to circumvent rate limiting. 
Our proof-of-concept implementation of this MitM attack demonstrates
its feasibility.

Second, structural protocol weaknesses involving inadequate cryptographic bindings
of ELTs, authentication cookies, and responses from the hazards database, 
prevent the system from detecting
if responses, within a TLS channel, 
from the hazards database were modified.
These inadequate bindings create a ``latent vulnerability'':
if a synthesizer were to reestablish a connection with the database
over the same TLS session, an adversary could replay and swap responses from the database without breaking TLS.
Although the current SecureDNA implementation does not allow such reconnections, 
the vulnerability might manifest under future implementations or TLS configurations.


Recently we made a responsible disclosure of our findings to the 
SecureDNA team, and we had fruitful discussions with them.
We learned that part of their (undocumented) security strategy 
is based significantly on monitoring, detecting, and responding
to certain types of malicious behaviors rather than on preventing such behaviors. 
They stated that, with the aid of source code that we have not seen,
their strategy includes automatic and manual intervention based on automatic alerts,
which aim to detect unexpected deviations from typical behaviors, whether
due to errors or malice. Their whitepaper does not mention this strategy.


Our contributions include:
(1)~A validated discernment of the SCEP, basic request, and exemption-handling protocols of the SecureDNA system
based in part on an examination of the source code.
(2)~CPSA models of these protocols.
(3)~Formal statements of security goals of these protocols, and 
formal-methods symbolic analysis of the models for structural weaknesses with respect to these goals.
(4)~Identification of structural flaws in SecureDNA's SCEP protocol,
showing that the protocol achieves only one-way authentication.
We also exhibit rate-limiting and \textit{denial-of-service (DoS)} 
attacks based on this vulnerability.
(5)~Identification of structural weaknesses in the SecureDNA protocols involving
inadequate authentication and cryptographic bindings, including
of ELTs, authentication cookies, and responses from the hazards database.
We also explain how an adversary could exploit this vulnerability 
to replay and swap responses from the database, without breaking TLS,
if a synthesizer were to reconnect with the database using the same TLS session.
(6)~Recommendations for strengthening the SecureDNA system, and formal-methods validation
of our suggested improvements to SCEP and 
to authentication and the bindings of protocol message components.

Section~\ref{sec:SCEP-analysis} presents our formal-methods analysis of SCEP, and
Section~\ref{sec:improvements} analyzes our improved SCEP+. 
Appendix~\ref{append:basic} analyzes the basic-query protocol; 
Appendix~\ref{append:exemption} analyzes the exemption-query protocol; and
Appendix~\ref{append:results} summarizes the results of our analyses in two tables.
All of our artifacts are available on GitHub~\cite{anonGitHub},
including our proof-of-concept implementation of our MitM attack
and complete CPSA input models and output shapes.


\section{Previous Work}
\label{sec:previous}

Aside from the abstract cryptographic studies by 
Baum et al.~\cite{baum2024CIC,baum2020cryptographic},
and a 2022 course project by Langenkamp 
et al.~\cite{MITArticle}, to our knowledge, our work
is the only security analysis of the SecureDNA system.
Whereas Baum et al.\ describe and analyze a cryptographic protocol for computing DOPRFs 
when some keyservers are malicious, we analyze the security of the
SecureDNA system design, engineering, and implementation, with a focus on
its query protocols. 
Baum et al.\ do not consider rate-limiting or DoS attacks.

Kane and Parker~\cite{kane2024screening} and
Hoffman et al.~\cite{hoffmann2023safety} review the landscape in DNA screening.
The U.S.\ Department of Health and Human Services~\cite{NIHguidelines}
offers non-binding security considerations, but
there are no standards or laws that require screening for dangers.
Founded in 2009, the International Gene Synthesis Consortium 
is a coalition of synthetic DNA providers and stakeholders
that use a common screening protocol~\cite{IGSCprotocol}.
Johns Hopkins~\cite{JHUhub} maintains a hub of useful information, including
a list of screening companies and tools.
We are not aware of any previous work that analyzes the design, security engineering, and implementation of any such system. Previous work deals mainly with biosafety, not cybersecurity.


\section{Background}  
\label{sec:background}

\thispagestyle{plain}
\pagestyle{plain}
We present brief relevant background regarding 
protocol analysis, strand spaces, and CPSA. For more details
about these topics, see~\cite{golaszewski2024limitations}.

\subsection{Protocol Analysis} 
\label{ssec:protocolanalysis}

\textit{Formal-methods analysis} of a protocol involves expressing the protocol in a formal mathematical model, stating propositions that reflect the protocol's desired security properties, and 
proving or disproving those propositions. Often, this process requires the assistance of specialized theorem-proving tools, such as
ProVerif~\cite{blanchet2016modeling}, Tamarin Prover~\cite{meier2013tamarin}, Maude-NPA~\cite{escobar2009maude}, or CPSA~\cite{liskov2016cryptographic}.
\textit{Symbolic}, versus computational, formal-methods analysis 
looks for \textit{structural} (fundamental logical) weaknesses, not cryptographic weaknesses.
Such formal-methods analysis of protocols will not detect 
implementation errors nor the application of protocols to inappropriate settings.

\subsection{Strand Spaces} 
\label{ssec:strands}

Strand spaces~\cite{fabrega1999strand} are a useful symbolic formalism for modeling the 
authentication and secrecy properties of cryptographic protocols. 
In the strand-space formalism, a cryptographic protocol is a set of \textit{roles} that form a template for legitimate {strands}. A \textit{strand} is a sequence of sent and received messages, where each message is an element of a term algebra that contains operations such as encryption and message concatenation. The \textit{strand space} for a cryptographic protocol is the set of all strands formed by term substitutions on the roles of the protocol or adversary strands. 

In the strand space formalism, executions of a protocol are modeled as \textit{bundles}.
A bundle for a protocol $\mathcal{P}$ is a set of strands (or prefixes of strands) from the strand space of $\mathcal{P}$ such that every reception node in the bundle corresponds to a unique transmission node in the bundle that sends the same message that was received.
The directed graph with edges connecting consecutive nodes on the same strand, and connecting corresponding reception and transmission nodes, must be acyclic so that the events in the bundle respect causality.
Bundles are of central importance in formal-methods analysis using strand spaces because each bundle provides an interpretation of a security goal formula. A protocol achieves a security goal if and only if the security goal is true under the interpretation of all of the protocol's bundles.


\subsection{CPSA} 
\label{ssec:cpsa}

CPSA~\cite{liskov2016cryptographic} is an open-source tool for analyzing cryptographic protocols within the strand-space model.
CPSA distinguishes itself as a \textit{model-finder}.
Its input is a model, which comprises strands consisting of roles, messages, variables, and a set of initial assumptions. When executing to completion, CPSA provably identifies all essentially different executions of the protocol within a 
\textit{Dolev-Yao (DY)} network intruder model~\cite{dolevyao} 
and outputs them as \textit{``shapes''}~\cite{liskov2011completeness}.
CPSA's model finding enables users to identify the strongest achieved security goal for an input model~\cite{rowe2016measuring}.
Users define CPSA models using LISP-like s-expressions that implement a custom language.
In these models, which superficially resemble (but are not) executable source code, users specify one or more roles, associated variables and messages, and assumptions.


\begin{figure*}[t] 
    \centering
    \includegraphics[width=\textwidth]{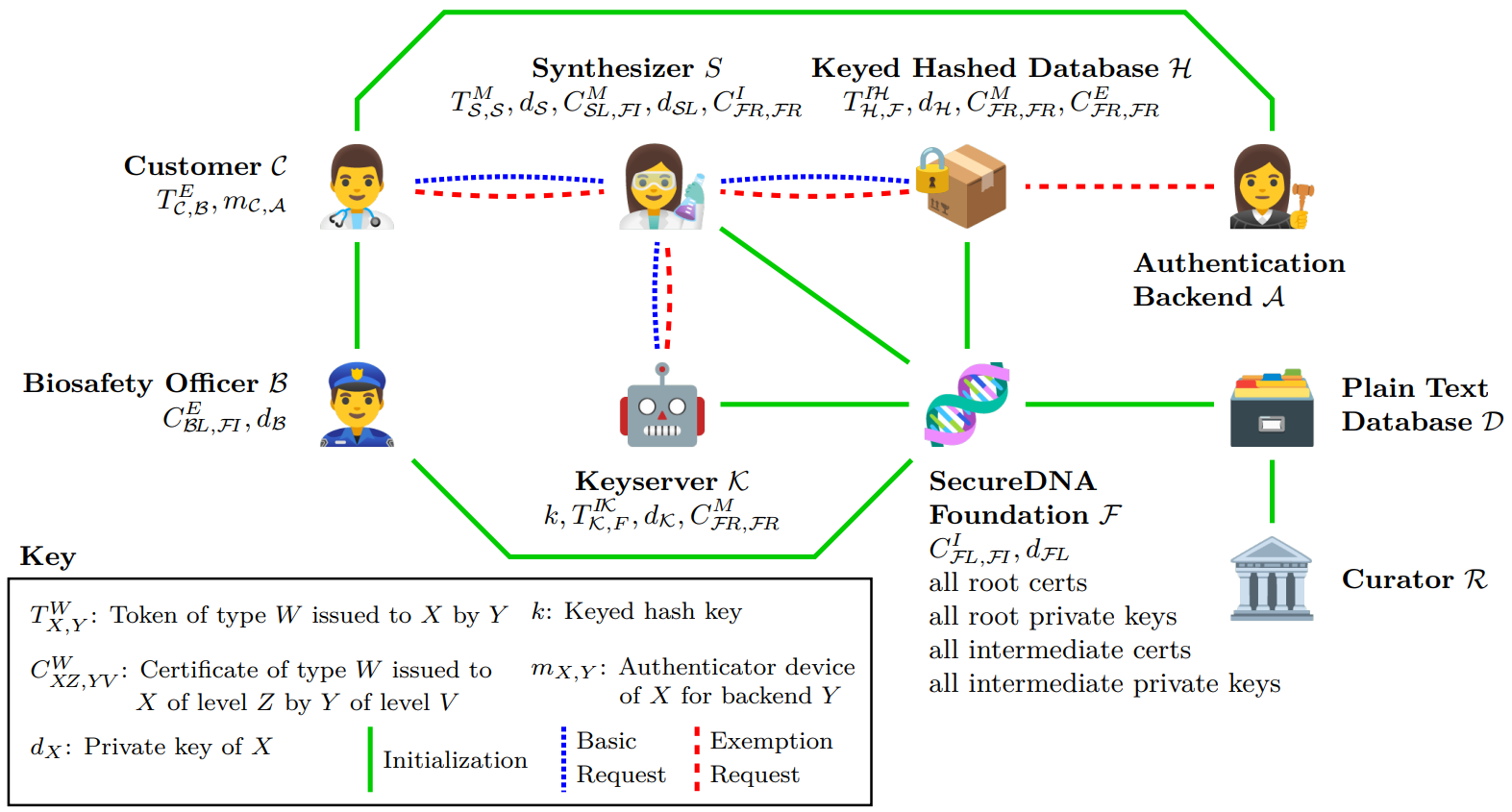}
    \caption{Architecture of the SecureDNA system showing the roles and cryptographic variables held by each role.}
    \label{fig:arch}
\end{figure*} 

\section{SecureDNA System}  
\label{sec:SDS}

The SecureDNA system~\cite{SecureDNA2024} provides a way for honest synthesizers to screen 
DNA synthesis-order requests for biohazards.  We summarize how this system works by explaining its
architecture, security goals, oblivious search for hazards,
certificate infrastructures, authentication tokens,
exemption-list tokens, and source code.

\subsection{Architecture} 
\label{ssec:arch}

As envisioned by its inventors, the SecureDNA system involves eight logical entities:
a \textit{Customer} {\myC} sends an order request, and
a \textit{Synthesizer} {\myS} [E10]\footnote{The notation ``[E10]'' means see Endnote~10 in Appendix~\ref{append:endnotes}.} receives and processes the order request.
There is a \textit{Plain Text Database} {\myD} of hazards, and
a \textit{Keyed Hashed Database} {\myH} [E7] of hazards.
A \textit{Curator} {\myR} populates {\myD}.
A \textit{Biosafety Officer} {\myB} grants exemption tokens to {\myC}.
Using Shamir secret sharing of a hash key $k$, 
a \textit{Distributed Keyserver} {\myK} [E4]
applies keys needed to populate and query {\myH}.
A \textit{SecureDNA Administrator} {\myF}---the SecureDNA Foundation---is the singular
root of trust that generates the hash key, 
establishes the root certificate for each of the three certificate hierarchies,
processes requests for synthesizer and exemption certificates, 
and releases system software and updates.
In addition, it is prudent to include
a ninth entity, an \textit{Authentication Backend} {\myA},
that verifies authentication requests for exemption tokens. See Figure~\ref{fig:arch}.

\subsection{Security Goals} 
\label{ssec:goals}

Based on the SecureDNA whitepaper~\cite{SecureDNA2024}, we understand that
the SecureDNA system aims to accomplish four security goals,
which we state informally:

\begin{itemize}
\item[SG1]
Keep the sequences in {\myD} secret to everyone except the Curator {\myR}.

\item[SG2]
Keep each synthesis order request $s$ secret to everyone except {\myC} and {\myS}.

\item[SG3]
Return to {\myS} a valid answer to whether a synthesis order request $s$ is in {\myD}.

\item[SG4]
When {\myC} presents an ELT to {\myS}, enable {\myS} to determine whether {\myC} 
is authorized to receive the synthesis of order request $s$, 
even if some sequences in $s$ are in {\myD}.
\end{itemize}

Neither the SecureDNA whitepaper~\cite{SecureDNA2024} nor
the initial technical cryptographic manuscript~\cite{baum2020cryptographic} 
states any of these goals precisely. The more recent cryptographic article~\cite{baum2024CIC}
deals only with the abstract mathematical cryptography and does not address SG3 or SG4.
Instead, this article focuses on ensuring correct operations of {\myK} when some of the
distributed keyservers are malicious.

To analyze these goals, it is essential to identify the adversarial model
(see Section~\ref{sec:adversary}). The whitepaper~\cite{SecureDNA2024} 
vaguely mentions the
\textit{semi-honest} and \textit{malicious} models without clearly articulating what
model is applied in what context.  The semi-honest model is a very weak model (participants must follow the protocol) and usually is insufficient for meaningful protection in many
network environments. Following well-accepted practice, we adopt the DY model.

Concerning SG1,  documentation states that
``secrecy [of the database of hazards] should be protected at the 
highest possible level while preserving usability,''~\cite[p.\ i]{baum2020cryptographic} and
``leakage of {\myH} about {\myD}  must be kept to a minimum.''~\cite[p.\ 17]{SecureDNA2024}.


However, when we presented our rate-limiting attack to the SecureDNA team, they 
asserted that SG1 is not very important, in part, because most known hazards are already described in publicly available documents. 
They explained that one important aspect of SG1 concerns
the relatively few newly discovered hazards
that are not yet widely known.
We consider SG1 an important goal;
much simpler designs would be possible if SG1 were eliminated.

Another important aspect of SG1 involves SecureDNA's
subtle biological strategy for
detecting functional mutants of hazards,
which are too numerous to list individually in {\myD}.
Every ``\textit{re-spinning}'' of {\myH} involves 
selecting new representative mutants to include in {\myD}, 
then \textit{rehashing} {\myH}.
To evade detection of mutants, it would be helpful to know
which mutants are in {\myD}.
An adversary might attempt a \textit{database-scraping attack} to
learn parts of {\myD} and to learn which mutants are in {\myD}.
Re-spinning complicates learning what mutants are in {\myD} by changing them, but
re-spinning does not prevent attacks that aim to learn the non-mutant entries
of {\myD}. Rehashing only makes such attacks less efficient. 
Although monitoring and re-spinning provide significant defenses, better protocols
would provide stronger protections.

We sensed that SecureDNA is extremely concerned about adoption, 
and to that end, their most important security goal is SG2. 
They are also very concerned about non-security goals, including 
speed and quality of service.


\subsection{Screening for Hazards} 
\label{ssec:screening}

The SecureDNA system screens for biohazards by performing a type of oblivious search
that uses a keyed hashed database {\myH} of hazards.
Because of its dangerous information, 
and to complicate evading detection of hazardous variants, 
the corresponding plain text database {\myD} must be kept secret.  
Sequences in {\myD} are relatively short (as small as 60~bits) 
to prevent the synthesis of longer hazards from shorter sequences.

The system depends in part on rate-limiting queries to {\myH} to prevent 
an adversary from searching a large number of the possible short sequences,
even if an exhaustive search of such sequences were computationally infeasible.
The system also periodically re-spins {\myH} to thwart the adversary from
learning which mutants are in {\myH} (see Section~\ref{ssec:goals}).

\subsection{Certificate Infrastructures} 
\label{ssec:PKI} 

The SecureDNA system manages a custom certificate-based 
\textit{public-key infrastructure (PKI)}
to support authentication and exemptions. There are three separate hierarchies: 
\textit{manufacturer} (to authenticate {\myS}), 
\textit{infrastructure} (to authenticate {\myK} and {\myH}), and 
\textit{exemption} (to prove {\myC} is authorized to synthesize some exempted sequences).
Each hierarchy has a separate root of trust created by the Administrator {\myF}:
\textit{manufacturer root} $\cert[M]({\mathcal F} R{\mathcal F}R)$, 
\textit{infrastructure root} $\cert[I]({\mathcal F}R{\mathcal F}R)$, and
\textit{exemption root} $\cert[E]({\mathcal F}R{\mathcal F}R)$.
SecureDNA distributes these roots of trust in its software release.

Each node in the hierarchy is a \textit{certificate}, which cryptographically binds an identity with its public key. 
Each non-root certificate is digitally signed by the private key
of its parent in the hierarchy. Each certificate has the following format:
\begin{equation}
    C_{x,y}^{W} = ( N_x, context, \{ M( N_x, context ) \}_{d_{y}} ),
\end{equation}

\noindent
where 
\begin{equation}
context = h,\sigma,p_x,N_y,p_y,{\delta}.
\end{equation}

\noindent
$x$ is the subject (receiver); $y$ is the issuer;
$M$ is a cryptographic hash function; and
$d_y$ is the private key of $y$.
For any data $v$ and any key $z$,
$\{v\}_z$ denotes encryption or signature of $v$ under key $z$.
Here, 
$h$ is a description of the certificate, which contains a version, 
type $W$ ({manufacturer}, {infrastructure}, {exemption}),
and hierarchy level ({root}, {intermediate,} {leaf}).
$\sigma$ is a randomly generated identifier assigned to the certificate, 
which identifies $x$.
$N_x$ and $N_y$ are the identifiers of $x$ and $y$, respectively, 
containing a name and email.
$p_x$ is the public key of $x$; 
$p_y$ is the public key of $y$; and
$\delta$ is a validity period, comprising a start and end date [E1,E2,E3]. 
In our notation $\cert[W](XZ,YV)$, $W$ is the type from $h$; 
$X$ is the subject; $Z$ is $X$'s hierarchy level; $Y$ is the issuer; 
and $V$ is the issuer's hierarchy level.


There are special types of certificates, called \textit{tokens}, 
including \textit{authentication tokens} and \textit{exemption-list tokens}.
The format of each token is similar to that of a certificate:

\begin{equation}
   T_{x,y}^{\pi} = ( \pi,u_\pi,context, \{ M( \pi,u_\pi, context ) \}_{d_{y}} ),
\end{equation}

\noindent where
$\pi$ is a certificate type, and
the contents of field $u_\pi$ depend on the certificate type $\pi$.

Each token is a leaf of a certificate chain. 
The usage of ``\textit{leaf}'' in the hierarchy level of certificates is a misnomer, as leaf certificates are intermediate certificates used to issue tokens. 
We will use SecureDNA's notation throughout the paper.
Each token is bundled with a chain of certificates that can be verified
to the corresponding root.
SecureDNA manages the creation and distribution of root certificates, 
creation of intermediate certificates, 
usage of intermediate certificates to sign leaf certificates, and 
revocation of certificates. 
Our notation for tokens is similar to that for certificates, 
but with no hierarchy level for subjects or issuers, 
because tokens have no hierarchy, and all issuers are leaf certificates.  

{\myF} provides {\myK} and {\myH} with a list of 
token identifiers $\sigma$ and public keys $p_x$ for revoked tokens.
The SecureDNA team explained that they did not incorporate standard revocation protocols for X.509 (such as OSCP~\cite{malpani2013x}) because of the additional burden on SecureDNA. Similarly, they did not empower {\myS} to act on revocation lists for {\myK} and {\myH}.
Instead, if {\myK} or {\myH} is compromised, they plan to re-key {\myH}.  
Whenever the root key is changed, SecureDNA would also have to issue a new client, because
the root certificates are embedded into the client software.

\begin{figure*}[t] 
    \centering
    \includegraphics[width=\textwidth]{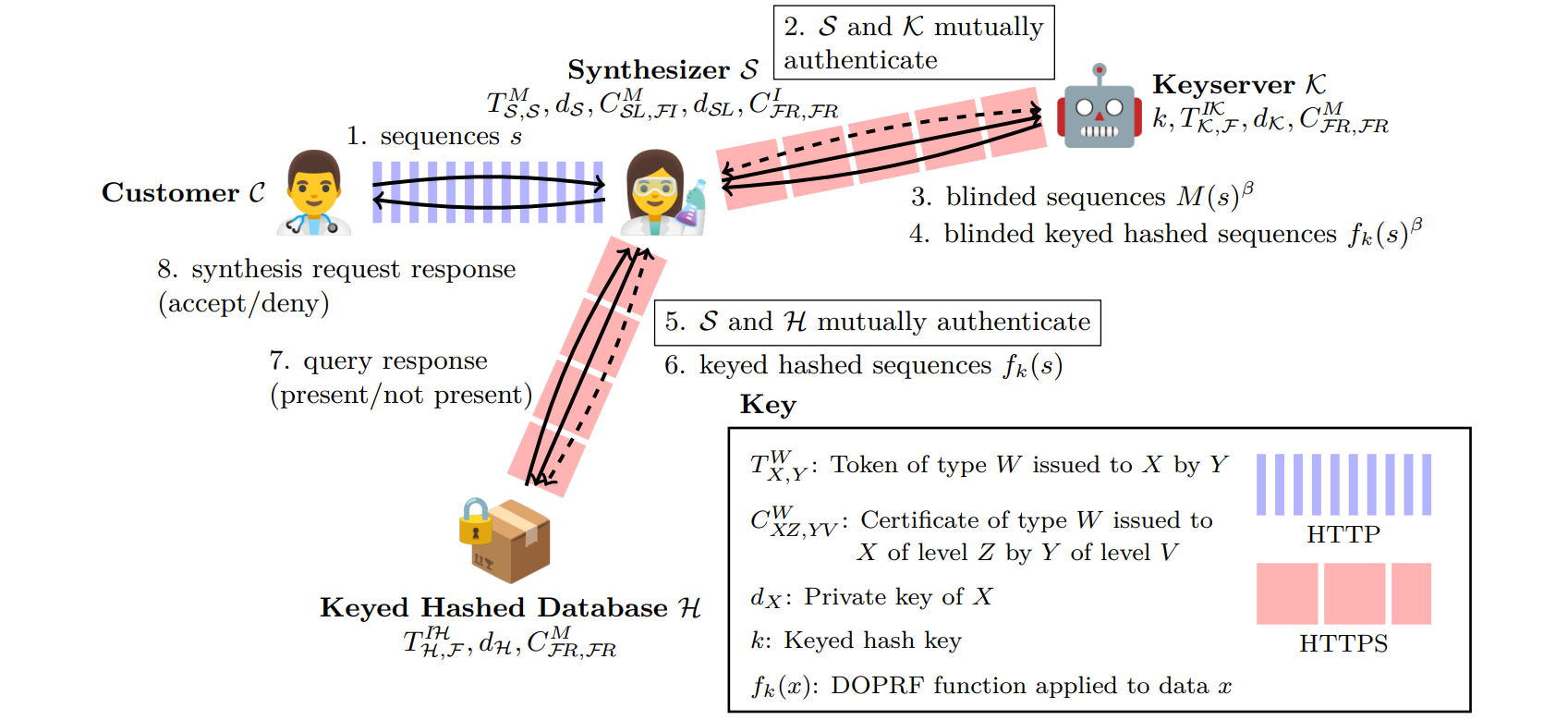}
    \caption{The basic order-request protocol.
    Steps~2 and 5 call a custom mutual authentication protocol SCEP 
    (see Section~\ref{ssec:custom}).}
    \label{fig:basic_proto}
\end{figure*} 


\begin{figure*}[t] 
    \centering
    \includegraphics[width=\textwidth]{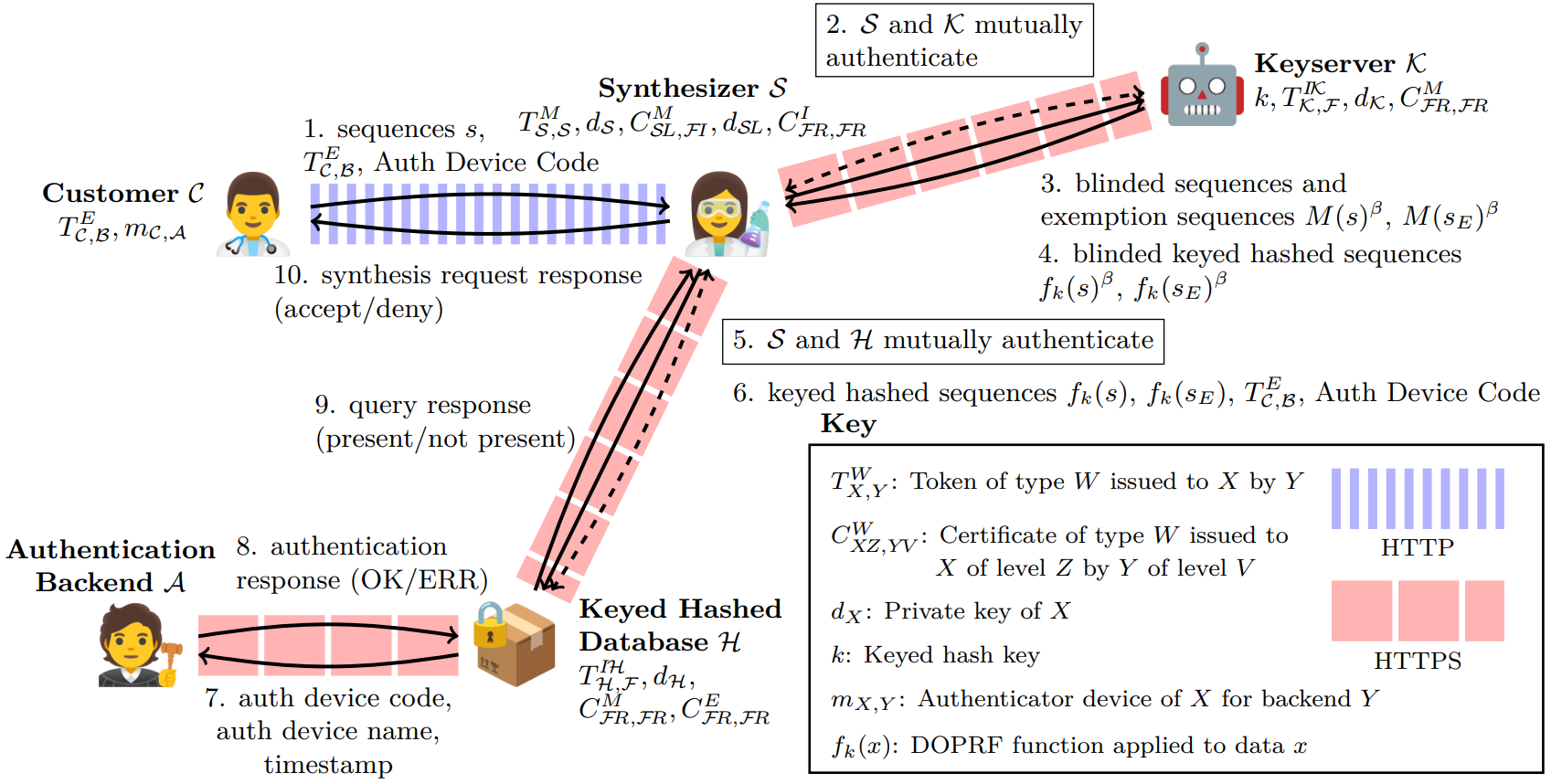}
    \caption{The exemption-handling protocol.  
    {\myC's} request includes an ELT signed by {\myB}
    and a one-time passcode (the ``Auth Device Code'').
    In {\myH}'s response to {\myS}, {\myH} sends a list of sequences that are in its database, 
    with boolean flags indicating which sequences are also in the exemption token. 
    {\myH} verifies that {\myC} included a valid passcode.
    Steps~2 and 5 call a custom mutual authentication protocol SCEP 
    (see Section~\ref{ssec:custom}).
    }
    \label{fig:exempt_proto}
\end{figure*} 


\subsection{Authentication Tokens} 
\label{ssec:authcookies}

SecureDNA uses three types of authentication tokens:
{\myS} uses \textit{Synthesizer tokens} $T_{x,y}^{M}$ to authenticate to {\myH} and {\myK} [E11,E12].
Each has a $u_{M}$ field containing an identifier for the synthesizer, and 
a rate limit $\mu$ for the synthesizer.
These tokens are bundled with chains rooted in $\cert[M]({\mathcal F}R{\mathcal F}R)$. 
Keyserver {\myK} uses \textit{Keyserver infrastructure tokens} $\token[I{\mathcal K}](xy)$ 
to authenticate to {\myS}. Each has a $u_{I\mathcal{K}}$ field containing {\myK}'s ID in the secret sharing scheme [E5,E6]. 
{\myH} uses \textit{Database infrastructure} tokens $\token[I{\mathcal H}](xy)$  
to authenticate to {\myS}. Each has an empty $u_{I\!\mathcal{H}}$ [E8,E9].
The keyserver and database infrastructure tokens are bundled with chains rooted in $\cert[I]({\mathcal F}R{\mathcal{F}}R)$.

\subsection{Exemption-List Tokens} 
\label{ssec:xtokens}

{\myC} uses \textit{exemption tokens} $T_{x,y}^{E}$  
to permit synthesis of dangerous sequences. 
Each has a $u_{E}$ containing a list of exempt sequences $s_E$, and 
an identifier $m_{{\mathcal C},A}$ for an \textit{authenticator device} 
(e.g., YubiKey) issued by {\myA} [E13,E14].
These tokens are bundled with chains rooted in $\cert[E]({\mathcal{F}}R{\mathcal{F}}R)$.
The public key $p_y$ in $T_{x,y}^{E}$ is optional, and is not used for authenticating token users. 
The only use of the key pair $(p_y,d_y)$ is to issue sub-tokens $T_{x,y}^{E'}$, 
which are copies of $T_{x,y}^{E}$ with the restriction that the
sequences in $u_{E'}$ must be a subset of the sequences in $u_{E}$.

\subsection{Source Code} 
\label{ssec:code}


The SecureDNA system is written mostly in the
Rust programming language~\cite{rust}, which enforces memory safety.
We analyzed Version 1.0.8 of the source code~\cite{sec_dna_github}, which comprises
approximately 64,000 lines of source code across nearly 300 files. 
SecureDNA adopts a TLS implementation from a Rust package.
To our knowledge there is no external documentation.
There are only sparse comments for the internals of the code and 
only sparse descriptions of user-facing functionality.
{\myF} distributes software using a trusted Linux package manager.

\section{SecureDNA Protocols} 
\label{sec:protocols}

We explain the two main protocols that the Synthesizer
{\myS} uses to screen order requests for hazards:
the basic order-request protocol and the exemption-handling protocol.
These protocols transform sequences $s$ using a keyed DOPRF $f_k(s) = M(s)^k$,
where $k$ is a key and $M$ is a cryptographic hash function.
These computations take place in a prime-order group in which the
Decisional \textit{Diffie-Hellman (DH)} Problem is hard~\cite{boneh1998decision}.
Figures~\ref{fig:basic_proto}--\ref{fig:exempt_proto} summarize how these protocols work;
for more details, see the associated Ladder 
Diagrams~\ref{fig:auth-ladder}--\ref{fig:two-ladders}.

When {\myS} communicates with the Keyserver or
the Hashed Database{\myH}, the
entities first establish authentication using a custom authentication sub-protocol
SCEP that involves exchanging nonces and authentication tokens 
(see Figures~\ref{fig:auth-ladder}--\ref{fig:two-ladders}).


\subsection{Basic Order-Request Protocol} 
\label{ssec:queryprotocol}

As shown in Figure \ref{fig:basic_proto}, the basic order-request protocol
begins with the Customer
{\myC} sending its synthesis request $s$ to {\myS}.
Synthesizer {\myS} blinds $s$ and sends it to {\myK} (actually, to several keyservers).
The keyservers apply their keyshares, and return the results to {\myS}. 
Then {\myS} combines the responses from the keyservers and unblinds the result to form the keyed hash of the request.
{\myS} sends this keyed hash to {\myH}. 
Next, {\myH} responds stating whether the keyed hash is in its database.
Finally, {\myS} reports back to {\myC} whether their synthesis request is allowed or denied.
Because the communication between {\myC} and {\myS} is intended to be hosted locally,
by default this communication occurs via HTTP;
optionally, HTTPS can be used.
Communication between {\myS} and {\myH}, and between {\myS} and {\myK}, 
occurs in a one-way authenticated TLS channel supported with certificates for {\myH} and {\myK}. At the beginning of these communications between 
{\myS} and {\myH}, and between {\myS} and {\myK}, the roles
complete a custom mutual authentication protocol SCEP 
(see Section~\ref{ssec:custom}).

\subsection{Exemption-Handling Protocol} 
\label{ssec:ELRprotocol} 

As shown in Figure~\ref{fig:exempt_proto},
the exemption protocol is similar to the basic order-request protocol. 
$C$ sends to {\myS} its synthesis request $s$, 
exemption-list token $T_{\mathcal{C},\mathcal{B}}^{E}$ 
obtained from the Biosafety Officer
{\myB}, and authenticator code $a$ (a one-time passcode)
from the authentication device listed in the exemption-list token.
{\myS} performs the same exchange involving {\myS} with {\myK} as in the basic order-request protocol. 
Then, {\myS} performs a second round with {\myK} to hash each of the sequences in $T_{\mathcal{C},\mathcal{B}}^{E}$. 
After {\myS} has assembled all keyed hashed sequences, 
{\myS} sends to {\myH} the keyed hash sequences, $T_{\mathcal{C},\mathcal{B}}^{E}$, and $a$. 
Then, {\myH} verifies the appropriateness of $T_{\mathcal{C},\mathcal{B}}^{E}$ by sending to the Authentication Backend
{\myA} the authenticator device name, $a$, and a timestamp. 
{\myA} responds with an OK or error. If {\myA} responds OK, 
{\myH} checks if the keyed hash of $s$ is in its database.
If it is, {\myH} checks if $s$ is in the exemption list. 
{\myH} then reports back to {\myS} if the keyed hash request is in its database, and 
if so, if it is in the exemption token. 
{\myS} reports back to {\myC} whether their synthesis request is granted.
The communication from {\myH} to {\myA} is an HTTPS request.
Communication between {\myS} and {\myH}, and between {\myS} and {\myK}, 
occurs in a one-way authenticated TLS channel supported with certificates for {\myH} and {\myK}.
At the beginning of these communications between 
{\myS} and {\myH}, and between {\myS} and {\myK}, the roles
complete a custom mutual authentication protocol SCEP 
(see Section~\ref{ssec:custom}).

\section{Adversarial Model}  
\label{sec:adversary}

We assume a malicious model in which a protocol communicant does not have to follow the protocol.
More specifically, we assume the DY model
in which the adversary has full control over all messages on the network and can manipulate an unbounded number of protocol sessions. The adversary can perform the roles of the legitimate
communicants. The adversary cannot break cryptographic primitives
but can compute keyed cryptographic primitives when the adversary knows the keys.
Objectives of the adversary include 
learning the sequences in the order request,
learning some or all plain text entries in the hazards database, and
causing the synthesizer to synthesize a sequence in the hazards database
(for which the adversary does not have permission to synthesize).

Inherent limitations of the system include that
the customer and synthesizer learn whether the order request is in 
the Database {\myD}, and
a malicious synthesizer can synthesize any sequence.


\section{Circumventing Rate Limiting} 
\label{sec:ratelimit}

SecureDNA includes \textit{rate-limiting mechanisms} that aim to 
make exhaustive ``dictionary attacks [on {\myD}] impossible'' by limiting the number of queries that can be made on {\myH} in a given time~\cite[p.\ 3]{SecureDNA2024}. 
We explain how rate limiting works and point out an attack that exploits
the failure of SCEP to provide mutual authentication.

The main stated purpose of rate limiting is to achieve SG1. The SecureDNA team explained
that another purpose is to detect software errors.  We note that rate limiting
can also try to defend against certain types of DoS attacks.
Without debating the significance of SG1 (see Section~\ref{ssec:goals}) and
without discussing technical biological issues involving 
the relationship of SG1 to determining functional variants of hazards, 
we will analyze SecureDNA's rate-limiting mechanisms.

\subsection{How Rate Limiting Works} 
\label{ssec:R-how}

SecureDNA includes two classes of rate-limiting mechanisms: 
(1)~software mechanisms that attempt to limit the number of times {\myK} executes the DOPRF,
which limits the number of queries {\myS} can make on {\myH}, and
(2)~auditing mechanisms that attempt to monitor, detect, and react to
unreasonable behaviors, implemented by 
non-public software (which we have not seen) and human intervention.

The software mechanisms work as follows.
Every time {\myS} executes the basic-request or exemption-handling protocol,
{\myK} and {\myH} separately record in their own query-databases
the time and number of sequences requested by {\myS}.
Entities {\myK} and {\myH} index their query-database by the 
identifier $\sigma$ in {\myS}'s authentication token {\tokS}. 
Entities {\myK} and {\myH} check if the total number of sequences requested by {\myS} 
within the previous 24 hours exceeds the rate limit $\mu$ in {\tokS} [E20].

The SecureDNA team explained that auditing mechanisms provide their main defense
against dictionary attacks and certain other malicious behaviors. They stated that
their reactive capabilities include the ability to revoke tokens, certificates,
and entire certificate chains.

\subsection{Potential Vulnerabilities} 
\label{ssec:R-vulnerabilities}

The failure of SCEP to provide mutual authentication is a serious vulnerability
for rate limiting and DoS: because neither {\myH} nor {\myK} know with whom they are
communicating, it is difficult to identify which actor is initiating a request.
The SecureDNA team stated that knowing the identity of the initiator is not necessary
because queries are associated with authentication tokens, not actors.  We will show, however, that it is possible for the adversary to steal and misuse 
authentication tokens.

Concerning the software mechanisms, the following facts might in some contexts be 
potential vulnerabilities [E21]:
(1)~{\myF} has no control over the rate limit $\mu$: 
{\myS} creates its own authentication token {\tokS}, setting its own $\mu$, 
which can be any 64-bit integer.  Note, however, that servers now limit
queries operationally (see Appendix~\ref{sec:other}). 
(2)~{\myF} has no control over the token identifier $\sigma$ in {\tokS}, which {\myS} chooses.
(3)~There is no limit on the number of authentication tokens {\myS} can create.
(4)~{\myS} can manufacture any number of valid leaf certificates.
(5)~The Auth Device Code is not bound to context, enabling a corrupt
{\myK} or {\myH} to steal and misuse it.
The SecureDNA team explained that they do not view these facts as vulnerabilities because
they guard against rate-limiting attacks primarily with auditing and re-spinning.
They also explained that, to avoid undue administrative burden, 
they do not want {\myF} to authorize tokens {\tokS} created by {\myS}.

\subsection{Attacks} 
\label{ssec:R-attacks}

(1)~We describe a rate-limiting attack that exploits the weakness that SCEP provides only one-way authentication.  In this attack, the adversary is a malicious {\myKp}, which is possible within SecureDNA's
adversarial model~\cite{baum2024CIC,baum2020cryptographic} and our DY model.
For example, {\myKp} might be the adversary performing the role of {\myK}
or it might be a compromised {\myK}.
An honest {\myS} connects with {\myKp}, which enables
{\myKp} to learn {\myS}'s authentication token without breaking TLS. 
This connection is possible if {\myS} trusts a compromised signing key included in {\myKp}'s certificate chain.
Then, following a strategy similar to that in Lowe's~\cite{Lowe1995} attack on 
NS, {\myKp} connects with an honest {\myK} masquerading as {\myS},
using {\myS}'s authentication token. 
After authenticating, the adversary {\myKp} now can issue as many queries
as they desire without detection, up to the rate limit for {\myS} and the
auditing limit. From the perspectives of {\myK} and auditors, 
the queries are associated with {\myS}.  
See Appendix~\ref{append:rate-attack} for details.

The adversary {\myKp} can repeat this attack exploiting different synthesizers without
detection by the software controls.  Eventually, the auditing controls would likely
notice a suspicious increase in the total volume of requests, even though they
would not necessarily know the cause of this increase.  The SecureDNA team stated
that their monitoring software would identify the offending server.

As a proof-of-concept, we implemented our MitM attack against a corrupt keyserver.  
The attack works as we envisioned.  For source code and pcap files, see our
artifacts~\cite{anonGitHub}.

A very similar rate-limiting attack is also possible involving a corrupt {\myHp},
who can masquerade as {\myS} to {\myH}. 
The attack can also be adapted as a DOS attack, using ideas from (2).

(2)~Before we learned of SecureDNA's auditing strategy, 
we pointed out that, if the system worked solely as defined in the publicly available software, then an adversary could carry out certain malicious activities without detection.
For example, due to intentional design and implementation features, 
anyone could circumvent query rate limits by
setting a high rate limit or
generating many additional nodes and tokens in the certificate hierarchy.
Also, an adversary could cause DoS by 
creating and misusing tokens with identifiers that 
collided with those of legitimate tokens, prompting SecureDNA
to revoke legitimate tokens (e.g., for exceeding rate limits).

We observe, however, that the adversary could not in other contexts
directly reuse a harvested authentication token 
because the adversary would also need its private key.
The adversary might not be able to reuse an ELT directly
because the ELT is also protected by a two-factor device~$m$.

\subsection{Risks} 
\label{ssec:R-risks}

The risk of Attack~(1) is low.  The adversary must trick synthesizers into
connecting with {\myKp}, for example, by compromising {\myKp} or 
manipulating the management of root certificates.
Eventually auditing will detect a suspicious increase in
total volume of queries, including queries associated with {\myKp} or {\myHp}; 
and there is limited value in learning {\myD}. 
Also, a malicious {\myKp} or {\myHp} could carry out other simpler attacks, such as
disabling rate limiting or returning incorrect answers.
Nevertheless, the adversary would likely be able to learn some entries of {\myD}.
It would be stronger security engineering to prevent this protocol-interaction
attack by requiring mutual authentication,
as Release 1.1.0 now does.

The SecureDNA team stated that Attacks~(2) would be detected and shut down
by their auditing mechanisms, and that collisions would be detected upon revocation.


\section{Weak Authentication and Inadequate Bindings} 
\label{sec:binding}

We identify structural weaknesses in SecureDNA's 
SCEP custom mutual authentication protocol and 
in bindings of ELTs, authentication tokens, and responses from {\myH}. 
These vulnerabilities violate the principle of defense in depth,
and the binding weakness permits a response-swapping attack 
if {\myS} were to  reconnect with {\myH} over the same TLS session 
(which the implementation disallows).
We begin by explaining SCEP.

\subsection{Server Connection Establishment Protocol (SCEP)}
\label{ssec:custom}

As shown in Figure~\ref{fig:auth-ladder}, 
in the basic-request and exemption-handling protocols,
when {\myS} communicates with {\myK} or {\myH}, 
rather than using \textit{mutual TLS (mTLS)}~\cite{campbell2020rfc}, 
SecureDNA engages in a {custom mutual authentication protocol}, 
called the \textit{Server Connection Establishment Protocol (SCEP)} [E15,E16,E17].
For a discussion of this design choice, see Section~\ref{sec:discussion}.
SecureDNA's documentation clearly reveals that SecureDNA intended for
SCEP to achieve mutual authentication---for example, the file containing the functionality is called ``mutual\_authentication.rs'' [E22].
SCEP begins by establishing a TLS channel, which
by default is based on one-way authentication.
During  SCEP, {\myS} receives a \textit{HTTPS request cookie} {\mycookie},
which {\myS} uses to authenticate future messages in the 
basic-request or exemption-handling protocol.
This cookie {\mycookie} is different from a SecureDNA authentication token.

\subsection{Vulnerabilities} 
\label{ssec:B-vulnerabilities}

As shown in Section~\ref{sec:SCEP-analysis}, SCEP does not achieve mutual authentication.
{\myS} authenticates the server ({\myH} or {\myK}), 
but the server does not know with whom it is communicating. 
The security properties of TLS with SCEP, including
with respect to authentication, 
are only slightly greater than those of TLS alone.
In terms of Lowe's authentication hierarchy~\cite{lowe_authentication}, 
using SCEP adds only the lowest guarantee---\textit{liveness}. 
This vulnerability stems from the omission of 
session parameters {\mycookie} and the server certificate from {\myS}'s signed response;
that is, the vulnerability stems from a failure to cryptographically bind the response to the context.


Similarly, in the basic-request and exemption-handling protocols, 
neither any response from {\myK} nor any response from {\myH} 
is cryptographically bound to {\myS}'s request.  
Consequently, an adversary might---in some situations---be able to replay
and swap these responses out of context, as we will next show.
The one-time passcode $a$ sent by {\myC} is not bound to the context
and hence can be reused by any malicious recipient.

\subsection{A Latent Response-Swapping Attack} 
\label{ssec:B-consequences}

SecureDNA's weak authentication and binding are undesirable 
characteristics that violate the principle of defense in depth.
Section ~\ref{sec:ratelimit} explains a rate-limiting attack 
made possible by SCEP's one-way authentication. 

We now explain how the inadequate binding of {\myH}'s responses might
permit a latent response-swapping attack without violating the TLS channel.
Suppose {\myS} \textit{reconnects} with {\myH} over the same TLS session (with the same keys).
The adversary could then potentially violate SG3 and SG4
by swapping a query response in the second connection 
with the corresponding one from the first connection.
This hypothetical attack does not assume a malicious {\myS}, {\myH}, or {\myK}.


To carry out this attack without detection, the adversary would have to deal with the following technical issues, which might depend in part on the TLS implementation.
First, the response swapping can be accomplished 
by manipulating bits at the TCP network layer, which is easy to do~\cite{HARRIS1999885}.
Second, the swapping must not change the TLS MAC value~\cite{rfc5246,rfc8446}, 
which depends in part on the sequence number.
For TLS 1.2 and 1.3, the sequence number is direction specific and
always begins with zero~\cite{rfc5246, rfc8446}. 
Thus, responses in the same locations from the two connections using the same keys will
have the same sequence numbers and thus can be swapped without changing the MAC value.

We observe, however, that if the adversary modified messages sent from {\myS} to {\myK}, 
the adversary would be detected given the ``active security'' measures in place 
for protecting against malicious {\myK}'s~\cite{baum2024CIC}.

\subsection{Risks} 
\label{ssec:B-risks}

The risk of the response-swapping attack depends in part on the likelihood 
of {\myS} reconnecting with {\myH} over the same TLS session. 
We understand that the SecureDNA implementation does not permit reconnections
or 0-RTTs, so this attack is not possible.
Specifically, by default, TLS 1.3 implementations set
reconnections and 0-RTT off [E23].
If, however, a future SecureDNA implementation allowed reconnections, 
the consequence could be that {\myS} takes incorrect actions, synthesizing hazards and denying legitimate synthesis requests. 
Although the overall risk might be relatively low, it
would be stronger security engineering to avoid this latent vulnerability.


The new Version 1.1.0 release includes 
an optional mode (which we have not studied)---called \textit{verifiable screening}---in which
{\myS} sends a hash of the query to {\myH}, and {\myH} returns to {\myS}
a signed hash of several items including 
the hash of the query, {\myH}'s response, and metadata.
Although the purpose is to enable {\myS} to save proof that it had screened the order, 
properly executed this strategy can also solve the binding weakness,
albeit inefficiently.

\section{Formal-Methods Analysis of SCEP}
\label{sec:SCEP-analysis}

We apply the strand-space formalism to analyze security properties 
of the custom SCEP protocol 
(see Section~\ref{ssec:custom}),
which executes within a one-way authenticated TLS session.
The purpose of SCEP is to mutually authenticate the synthesizer {\myS} 
with the infrastructure system {\myK} or {\myH}.
To carry out our analysis, we
(1)~define a \textit{strand space} (Definition~\ref{def:strand-space}) for the combined TLS-SCEP protocol,
(2)~formalize the SCEP mutual authentication security goal as a pair of logical formulas describing properties of distinct execution models on the SCEP strand space, and
(3)~use CPSA to verify the formal security goals.
Our CPSA inputs and outputs are available on GitHub~\cite{anonGitHub}.

It is useful to introduce the concept of a \textit{trace}, which
is a finite, non-empty sequence of events (message transmissions or receptions)
that take place within a designated channel.

\begin{definition}[Strand Space]
\label{def:strand-space}
    A \textit{directed term} is a pair $(d, t)$, where $t \in A$ ($A$ is the set of all possible protocol terms, or messages) and $d$ is either incoming ($-$) or outgoing ($+$).
    A \textit{strand} $s$ comprises a trace of directed terms, where $s$ specifies the actions of a legitimate party or of the adversary.
    A \textit{strand space} is a set $\Sigma$ of strands.
\end{definition}

\subsection{SCEP Strand Space}
\label{ssec:scep-strand-space}

To define the SCEP strand space (Definition~\ref{def:scep-strand-space}), we define \textit{penetrator strands} that model adversarial behavior and \textit{regular strands} that model behavior of legitimate protocol roles.
Definition~\ref{def:pen-strands} specifies the set $\Sigma_{pen}$ of penetrator strands.
Definition~\ref{def:scep-regular-strands} specifies sets of regular strands that carry out the roles of the synthesizer $\Sigma_{SCEP-\mathcal{SR}}$ and the responder $\Sigma_{SCEP-\mathcal{RS}}$.
We use the SCEP strand space, which combines the penetrator and regular strands, to model protocol executions that potentially incorporate adversarial behavior.

\begin{definition}[Penetrator Strands]
\label{def:pen-strands}
    Let $T \subseteq A$ be a set of atoms that represent principals, string literals, keys, and nonces.
    Let $K \subseteq T$ be a set of all encryption keys, including inverse keys.
    Let $e(t, k)$ be the encryption of term $t$ under key $k$.
    Depending on whether $k$ is a symmetric or asymmetric key, $e(t, k)$ may denote either symmetric or asymmetric encryption; $k$ is a symmetric key when $k^{-1} = k$.
    The set $\Sigma_{pen}$ is a set of all strands $p$, where $p$ has one of the penetrator traces below:
    \begin{enumerate}
        \item $Generate(t)$: ${[}+t{]}$ where $t \in T$.
        \item $Encrypt(g, k)$: ${[}-g,-k,+e(g,k){]}$ where $g \in A, k \in K$.
        \item $Decrypt(e(g, k), k^{-1})$: ${[}-e(g, k), -k^{-1}, +g{]}$ where $g \in A, k^{-1} \in K$.
        \item $Hash(g)$: ${[}-g, +h(g){]}$ where $g \in A$.
        \item $Concatenate(g, h)$: ${[}-g, -h, +g||h{]}$ where $g,h \in A$.
        \item $Separate(g||h)$: ${[}-g||h, +g,+h{]}$ where $g, h \in A$.
    \end{enumerate}
\end{definition}

Because SCEP depends on an existing one-way authenticated TLS connection, we provide and incorporate traces for TLS~1.2 with ephemeral DH key-exchange (Definition~\ref{def:tls-trace}).
SCEP supports TLS~1.2 and TLS~1.3 channels, specifying that TLS~1.2 must use an ephemeral DH ciphersuite.
Based on a previous analysis~\cite{SBP}, traces of TLS~1.2 with ephemeral DH and TLS~1.3 are cryptographically equivalent.
SCEP and the underlying TLS connection depend on asymmetric cryptography.
In Definition~\ref{def:pubprivk}, the function $pk$
expresses a principal's public key, and $sk$ expresses the corresponding private key.

\begin{definition}[Public and Private Keys]
\label{def:pubprivk}
Let $pk: Name \to K$ and $sk: Name \to K$ be one-to-one functions with disjoint images such that $pk(a) = sk(a)^{-1}$ for all $a \in Name$.
\end{definition}

\begin{definition}[TLS Traces]
\label{def:tls-trace}
We define traces that establish a TLS handshake between a client $C$ and a server $S$. First, we define the terms we use within the trace:
\begin{itemize}
    \item $C, S, CA \in Name$
    \item $r_C, r_S \in Nonce$
    \item $e_S, e_C \in Expt$
    \item $Cert_S: S||pk(S)||e(h(s, g^{e_S}), sk(CA)))$
    \item $S_{KE}: g^{e_S}||e(h(r_C || r_S || g^{e_S}), sk(S))$
    \item $PMS: g^{e_S \times e_C}$
    \item $C_{WRITE}: h(PMS, r_C, r_S, \text{``client-write''})$
    \item $S_{WRITE}: h(PMS, r_C, r_S, \text{``server-write''})$
    \item $C_{FMESG}: r_C || r_S ||Cert_S||S_{KE}$
    \item $C_{FIN}: \\e(h(PMS || \text{``client-fin''} || C_{FMESG}),C_{WRITE})$
    \item $S_{FIN}:\\
    e(h(PMS || \text{``server-fin"} ||C_{FMESG} ||C_{FIN}), S_{WRITE}$).
\end{itemize}

Let trace \\$Tr_{TLS-C}(C, S, CA, r_C, r_S, e_S, e_C) =$
\begin{enumerate}
    \item $+r_C$
    \item $-r_S||Cert_S||S_{KE}$
    \item $+g^{ec}||C_{FIN}$
    \item $-S_{FIN}$.
\end{enumerate}

Let the complementary trace \\$Tr_{TLS-S}(C, S, CA, r_C, r_S, e_S, e_C) =$
\begin{enumerate}
    \item $-r_C$
    \item $+r_S||Cert_S||S_{KE}$
    \item $-g^{ec}||C_{FIN}$
    \item $+S_{FIN}$.
\end{enumerate}
\end{definition}

For the SCEP strand space, we specify a simple token (Definition~\ref{def:token}) 
that acts as a certificate of a principal's public asymmetric key.
By contrast, the system implementation tokens embed custom SecureDNA certificates and corresponding signatures.
From our analysis, this simplification provides equivalent security guarantees (under the assumption that the SecureDNA foundation is an honest \textit{certificate authority (CA)}), 
while simplifying execution models resulting from the SCEP strand space.

\begin{definition}[SCEP Token]
\label{def:token}
Let $X,Y \in Name$, $data \in Text$.

\noindent Let $Tok_{X,Y} = \text{``Token''} ||X||pk(X)||Y||pk(Y)||data$.

\noindent We define an SCEP \textit{token} as a compound term with the structure

{$T_{X, Y} = Tok_{X,Y}||e(h(Tok_{X,Y}), sk(Y))$}.
\end{definition}

\begin{definition}[SCEP Traces]
\label{def:scep-traces}
We define traces that carry out SCEP between {\myS} and a responding infrastructure server {\myK} or {\myH}, which we define as {\myW}. As before, we first define the terms that we use within the traces:
\begin{itemize}
    \item $\mathcal{S}, \mathcal{W}, M, CA \in Name$
    \item $r_\mathcal{S}, r_\mathcal{W}, r'_\mathcal{S}, r'_\mathcal{W}, \omega\in Nonce$
    \item $e_\mathcal{S}, e_\mathcal{W} \in Expt$
    \item $T_{S,M}$ (Definition~\ref{def:token})
    \item $T_{\mathcal{W,CA}}$ (Definition~\ref{def:token})
    \item $\mathcal{S}_{WRITE}: h(g^{e_\mathcal{S} \times e_\mathcal{W}}, r'_\mathcal{S}, r'_\mathcal{W}, \text{``client-write''})$
    \item $\mathcal{W}_{WRITE}: h(g^{e_\mathcal{S} \times e_\mathcal{W}}, r'_\mathcal{S}, r'_\mathcal{W}, \text{``server-write''})$.
\end{itemize}

\noindent
Let $Tr_{SCEP-\mathcal{S}\mathcal{W}}(\mathcal{S}, \mathcal{W}, M, CA, r_\mathcal{S}, r_\mathcal{W}, r'_\mathcal{S}, r'_\mathcal{W}, \omega, e_\mathcal{S}, e_\mathcal{W}):$

\begin{enumerate}
    \item $Tr_{TLS-C}(\mathcal{S},\mathcal{W},CA,r'_\mathcal{S},r'_\mathcal{W}, e_\mathcal{S}, e_\mathcal{W})$
    (Definition~\ref{def:tls-trace})
    \item $+e(r_\mathcal{S}||T_{S,M}, \mathcal{S}_{WRITE})$
    \item $-e(\omega || r_{\mathcal{W}} || T_{\mathcal{W,CA}} \\||e(h(\text{``server-mutauth"}, r_{\mathcal{S}}, r_{\mathcal{W}},T_{\mathcal{W,CA}}), pk(\mathcal{W})),\\\ \mathcal{W}_{WRITE})$
    \item $+e(\omega\\||e(h(\text{``client-mutauth"}, r_{\mathcal{S}}, r_{\mathcal{W}},T_{S,M}), pk(\mathcal{S})),\\\ \mathcal{S}_{WRITE})$.
\end{enumerate}

Let $Tr_{SCEP-\mathcal{W}\mathcal{S}}$ be complementary to $Tr_{SCEP-\mathcal{S}\mathcal{W}}$, such that $Tr_{SCEP-\mathcal{W}\mathcal{S}}$ (1)~includes $Tr_{TLS-S}$ rather than $Tr_{TLS-C}$, and
(2)~inverts the directions of terms 2, 3, and 4.
\end{definition}

\begin{definition}
\label{def:scep-regular-strands}
We define two sets of ``regular'' strands that carry out the SCEP protocol:
\begin{enumerate}
    \item Let $\Sigma_{SCEP-\mathcal{SW}}$ be the set of all strands with a trace of the form $Tr_{SCEP-\mathcal{S}\mathcal{W}}$.
    \item Let $\Sigma_{SCEP-\mathcal{WS}}$ be the set of all strands with a trace of the form $Tr_{SCEP-\mathcal{W}\mathcal{S}}$.
\end{enumerate}
\end{definition}

\begin{definition}[SCEP Strand Space]
\label{def:scep-strand-space}
The SCEP strand space comprises $\Sigma_{pen} \cup \Sigma_{SCEP-\mathcal{SW}} \cup \Sigma_{SCEP-\mathcal{WS}}$.
\end{definition}

\subsection{Security Goals}
\label{ssec:scep-security-goals}

We state formal security goals (Definitions~\ref{def:scep-sw-secure}, \ref{def:scep-ws-secure}) for SCEP.
To begin, we first formalize confidentiality (Definition~\ref{def:goal_conf}) and agreement (Definition~\ref{def:goal_agree}).
These high-level security goals assert that, within a given strand space $\mathcal{P}$ and under an explicit set of \textit{origination} assumptions, there exist no execution models (i.e., \textit{shapes}), for which the specified properties do not hold.
Later, we compose these definitions of confidentiality and agreement to define specific security goals for SCEP and the SecureDNA query and exemption protocols.

To prove confidentiality properties of execution models within a strand space $\mathcal{P}$, we define \textit{listener strands} (Definition~\ref{def:listener-strand}).
To test the confidentiality of a sensitive value $x$, we define a listener strand with trace $Tr_\mathcal{L}(x)$.
An execution model of $\mathcal{P}$ that satisfies $Tr_\mathcal{L}(x)$ illustrates a manner in which $x$ leaks to the adversary.
Thus, to prove confidentiality, we show that there exists no listener strand $Tr_\mathcal{L}(x)$ in any complete execution of $\mathcal{P}$.

\begin{definition}[Listener Strands]
\label{def:listener-strand}
A \textit{listener strand} for $t$ is any strand with trace $Tr_\mathcal{L}(t) = {[}-t, +t {]}$, where $t \in A$ is an arbitrary term. Let $\Sigma_\mathcal{L}$ be the set of all listener strands.
\end{definition}

\begin{definition}[Confidentiality]
\label{def:goal_conf}
    Let $Tr_\mathcal{X}$ be any trace in $\mathcal{P}$.
    Let $t$ be any term for which to test confidentiality.
    Let $\mathbf{orig}$ be the conjunction of origination predicates.
    Let $\vec{V}$ be a list of parameters of trace $Tr_\mathcal{X}$, such that the term $t$ is an element of $\vec{V}$.
    For any strand $r \in \mathcal{P}$, $r:Tr$ indicates that $r$ has trace $Tr$.
    Then,
    
    \noindent $Conf_{\mathcal{P}}(Tr_\mathcal{X},t,\mathbf{orig}) \iffdef \\ (\forall r,l\in\mathcal{P}\cup\Sigma_\mathcal{L},r:Tr_\mathcal{X}(\vec{V}) \wedge l:Tr_\mathcal{L}(t) \wedge \mathbf{orig} \implies \bot)$.

\end{definition}

\begin{definition}[Agreement]
\label{def:goal_agree}
    Let $Tr_\mathcal{X}$ be any trace in $\mathcal{P}$.
    Let $Tr_\mathcal{Y}$ be any trace in $\mathcal{P}$ such that $Tr_\mathcal{X} \not= Tr_\mathcal{Y}$.
    Let $\vec{C}$ be any list of values for which to test agreement.
    Let  $\vec{V}_\mathcal{X}$ be a list of parameters of the trace $Tr_\mathcal{X}$ and  $\vec{V}_\mathcal{Y}$ be a list of the parameters of trace $Tr_\mathcal{Y}$.
    Let $\mathbf{orig}$ be any conjunction of origination predicates.
    Then,

    \noindent $Agree_\mathcal{P}(Tr_\mathcal{X}, Tr_\mathcal{Y}, \vec{C}, \mathbf{orig}) \iffdef \\ (\forall r \in \mathcal{P}, r:Tr_\mathcal{X}(\vec{V_\mathcal{X}} \cup \vec{C}) \wedge \mathbf{orig} \implies \\ \exists s \in \mathcal{P}, s:Tr_\mathcal{Y}(\vec{V_\mathcal{Y}} \cup \vec{C}))$.
\end{definition}

\begin{definition}[SCEP Terms]
Let $\vec{T}$ be the list of terms:
\begin{enumerate}
    \item $\mathcal{S}, \mathcal{W}, M, CA \in Name$
    \item $r_\mathcal{S}, r_\mathcal{W}, r'_\mathcal{S}, r'_\mathcal{W}, \omega \in Nonce$
    \item $e_\mathcal{S}, e_\mathcal{W} \in Expt$.
\end{enumerate}
\end{definition}

We formalize security goals of SCEP from the perspective of {\myS} (Definition~\ref{def:scep-sw-secure}) and from the perspective of {\myW} (Definition~\ref{def:scep-ws-secure}).
Each of these definitions specifies origination assumptions of the cryptographic perspective of {\myS} or {\myW}, and comprises a conjunction of subgoals: (1)~confidentality of $\omega$, and (2)~agreement on the identities {\myS} and {\myW}, the corresponding nonces $r_{\mathcal{S}}, r_\mathcal{W}$, and $\omega$.
Should both goals hold for all unique execution models in the SCEP strand space, we prove that $\omega$ does not leak to the adversary and that {\myS} and {\myW} achieve 
\textit{injective agreement}~\cite{lowe_authentication}
on the SCEP session parameters.
Because of a crucial weakness in the SCEP (Section~\ref{ssec:scep-analysis}), we find SCEP is not ``{\myW}-{\myS}'' secure.

\begin{definition}[SCEP $\mathcal{S}-\mathcal{W}$ Secure]
\label{def:scep-sw-secure}
Let $\mathbf{orig}$ be the conjunction of the following origination assumptions on $\vec{T}$:
\begin{enumerate}
    \item $Non(sk(\mathcal{S}))$, $Non(sk(\mathcal{W}))$, $Non(sk(M))$, $Non(sk(CA))$, where $Non$ 
    indicates non-origination, and
    \item $Uniq(\omega)$, $Uniq(r_\mathcal{S})$, $Uniq(r'_\mathcal{S})$, $Uniq(e_\mathcal{S})$, $Uniq(e_\mathcal{W}))$, where $Uniq$ indicates unique origination.
\end{enumerate}

\textit{Uniquely originating} terms are unknown to protocol participants and the network until a legitimate strand emits them as part of a message, enabling us to model random nonces, fresh secret keys, and other values that must be unique for each execution of a protocol.
\textit{Non-Originating} values are values such as private keys, which the adversary does not know, cannot guess, and will never appear on the network in a decryptable form.

An SCEP strand space $\mathcal{P}$ is $\mathcal{S}-\mathcal{W}$ secure 
\textit{if and only if (iff)} the conjunction of the 
following security goals holds:
\begin{enumerate}
    \item $Conf_\mathcal{P}(Tr_{SCEP-\mathcal{SW}}(\vec{T}), \omega, \mathbf{orig})$, and
    \item $Agree_\mathcal{P}(\\Tr_{SCEP-\mathcal{SW}}(\vec{T}), Tr_{SCEP-\mathcal{WS}}(\vec{T})$,  \
    ${[\mathcal{S,W},r_\mathcal{S}, r_\mathcal{W}, \omega {]}, \mathbf{orig})}$.
\end{enumerate}
\end{definition}

\begin{definition}[SCEP $\mathcal{W}-\mathcal{S}$ Secure]
\label{def:scep-ws-secure}
Let $\mathbf{orig}$ be the conjunction of the following origination assumptions on $\vec{T}$:
\begin{enumerate}
    \item $Non(sk(\mathcal{S}))$, $Non(sk(\mathcal{W}))$, $Non(sk(M))$, $Non(sk(CA))$, where $Non$ indicates non-origination, and
    \item $Uniq(\omega)$, $Uniq(r_\mathcal{W})$, $Uniq(r'_\mathcal{W})$, $Uniq(e_\mathcal{S})$, $Uniq(e_\mathcal{W}))$, where $Uniq$ indicates unique origination.
\end{enumerate}

An SCEP strand space $\mathcal{P}$ is $\mathcal{W}-\mathcal{S}$ secure iff the conjunction of the following security goals holds:
\begin{enumerate}
    \item $Conf_\mathcal{P}(Tr_{SCEP-\mathcal{WS}}(\vec{T}), \omega, \mathbf{orig})$, and
    \item $Agree_\mathcal{P}(Tr_{SCEP-\mathcal{WS}}(\vec{T}),\\ Tr_{SCEP-\mathcal{SW}}(\vec{T}), {[\mathcal{S,W},r_\mathcal{S}, r_\mathcal{W}, \omega {]}, \mathbf{orig})}$.
\end{enumerate}
\end{definition}

\subsection{Analysis}
\label{ssec:scep-analysis}

Using CPSA, we describe \textit{all} minimal, essentially different execution models (shapes) within a strand space under a specific security goal's origination assumptions.
Theorems~\ref{thrm:scep-s-w} and~\ref{thrm:scep-w-s} prove assertions about these SCEP security goals.
A security goal holds iff it is true for all shapes possible within a strand space, which
CPSA might verify by exhaustive search.
A single counterexample is sufficient to disprove a security goal, when CPSA constructs a potential attack.

\begin{theorem}[SCEP $\mathcal{S}-\mathcal{W}$ Security]
\label{thrm:scep-s-w}
If $\mathcal{P}$ is an SCEP strand space, then $\mathcal{P}$ is $\mathcal{S}-\mathcal{W}$ secure.
\end{theorem}

\begin{proof}[Proof by Enumeration]
Carrying out an exhaustive search, CPSA enumerates all 
essentially different (unique) shapes under the $\mathcal{S}-\mathcal{W}$ security goal assumptions.
None of these shapes contradict the security goal.
\end{proof}

Theorem~\ref{thrm:scep-s-w} holds because {\myS} establishes a confidental and authenticated TLS channel with {\myW}.
As a result, the synthesizer believes $\omega$ to be confidential and agrees on $\omega$ and the nonces $r_\mathcal{S}, r_\mathcal{W}$ with the responder.

\begin{theorem}[SCEP $\mathcal{W}-\mathcal{S}$ Security]
\label{thrm:scep-w-s}
If $\mathcal{P}$ is an SCEP strand space, then $\mathcal{P}$ is \textbf{not} $\mathcal{W}-\mathcal{S}$ secure.
\end{theorem}

\begin{proof}[Proof by Counterexample]
Using CPSA, we identify all unique shapes that contradict the confidentiality and agreement subgoals.

\textit{Confidentiality.} In each shape, a listener strand learns the value of $\omega$ when $\mathcal{W}$ communicates with penetrator strands masquerading as a legitimate client.

\textit{Agreement.}
In each shape, a legitimate strand with the trace of $\mathcal{W}$ fails to agree on the values of $\omega, r_\mathcal{S}, r_\mathcal{W}$ with any regular strand.
\end{proof}

In counterexamples that prove Theorem~\ref{thrm:scep-w-s}, penetrator strands establish parallel TLS connections with the legitimate $\mathcal{W}$ and {\myS}.
Consequently, an adversary is able to break the confidentiality and agreement subgoals.
This failure results from (1)~the one-way authentication of TLS, which only authenticates the server to the client, and (2)~SCEP's failure to mutually authenticate {\myS} and {\myW}.

SCEP fails to mutually authenticate because it fails to make explicit the intended recipients of Messages~(3) and (4) in the trace (Definition~\ref{def:scep-traces}), enabling a MitM attack.
This flaw is similar to that of the 1978 NS public-key protocol.
Due to this weakness, composing SCEP with TLS produces an outcome no more secure than one-way authenticated TLS, which fails to satisfy the SCEP goal of mutual authentication.
In Section~\ref{sec:improvements}, we propose and formally verify improvements to the SCEP that enable mutual authentication, and thus $\mathcal{W}-\mathcal{S}$ security.



\section{Suggested Improvements, Including to SCEP}
\label{sec:improvements}

We propose and verify a correction to the SCEP that satisfies our formal security goals (Section~\ref{ssec:scep-security-goals}).
In summary, the SCEP fails to mutually authenticate because it does not bind critical values (token of the other party, request cookie) to the hash that each communicant signs and transmits.
The result is that {\myS} is certain of the identity of {\myW} (because SCEP runs within a one-way authenticated TLS session), but {\myW} has no guarantee that the cookie $\omega$ remains confidential, and {\myW} cannot determine that it communicates with any legitimate instance of {\myS}.
To correct this error, we modify the SCEP traces in Section~\ref{ssec:scep-strand-space} and repeat our analysis on the resulting strand space for our improvement,
which we call SCEP+.

First, we update the SCEP traces (Definition~\ref{def:scep-traces}) to include additional information in the hashes by incorporating both communicant tokens and the cookie into each hash (Definition~\ref{def:scep-plus-traces}).
Each party signs its hash using a key unknown to the adversary. This step binds $\omega$ and the tokens $T^M_{\mathcal{S},M}, T^{IF}_{\mathcal{W,CA}}$ within a single SCEP session.


\begin{definition}
[SCEP+ Traces]
\label{def:scep-plus-traces}
Let $Tr_{SCEP+-\mathcal{SW}}$ and $Tr_{SCEP+-\mathcal{WS}}$ modify the traces $Tr_{SCEP-\mathcal{SW}}$ and $Tr_{SCEP-\mathcal{WS}}$ (Definition~\ref{def:scep-traces}) with the alterations:
\begin{enumerate}
    \item $h(\text{``server-mutauth"}, r_{\mathcal{S}}, r_{\mathcal{W}},\omega, T_{\mathcal{S},M}, T_{\mathcal{W,CA}})$ 
    \item $h(\text{``client-mutauth"}, r_{\mathcal{S}}, r_{\mathcal{W}}, \omega, T_{\mathcal{S},M}, T_{\mathcal{W,CA}})$.
\end{enumerate}
\end{definition}


Next, we update the definitions of SCEP $\mathcal{S-W}$ (Definition~\ref{def:scep-sw-secure}) and $\mathcal{W-S}$ (Definition~\ref{def:scep-ws-secure}) security with the $SCEP+$ traces.
Because we are not adding any terms, the security goals do not change.
We refer to these security goals as $SCEP+$ $\mathcal{S-W}$ and $SCEP+$ $\mathcal{W-S}$.

Finally, we state and prove two new theorems that assert confidentiality of $\omega$ and agreement on the session parameters for SCEP+.
Because SCEP+ includes sufficient encrypted information in each message, both parties are able to authenticate each other.
As a result, we prove 
Theorem~\ref{thrm:scep-plus-s-w} and Theorem~\ref{thrm:scep-plus-w-s} 
using CPSA.

\begin{theorem}[SCEP+ $\mathcal{S}-\mathcal{W}$ Security]
\label{thrm:scep-plus-s-w}
If $\mathcal{P}$ is an SCEP+ strand space, then $\mathcal{P}$ is $\mathcal{S}-\mathcal{W}$ secure.
\end{theorem}

\begin{proof}[Proof by Enumeration]
There exists no shape that contradicts $\mathcal{S}-\mathcal{W}$ security.
\end{proof}

\begin{theorem}[SCEP+ $\mathcal{W}-\mathcal{S}$ Security]
\label{thrm:scep-plus-w-s}
If $\mathcal{P}$ is an SCEP+ strand space, then $\mathcal{P}$ is $\mathcal{W}-\mathcal{S}$ secure.
\end{theorem}

\begin{proof}[Proof by Enumeration]
There exists no shape that contradicts $\mathcal{W}-\mathcal{S}$ security.
\end{proof}

Our improvement, SCEP+, does not require any additional cryptographic calls and requires minimal changes.
SCEP+ achieves SCEP's goal of mutual authentication.
Software release 1.1.0 implements our recommended SCEP+ protocol, 
which change involved approximately five lines of code.



\section{Discussion}  
\label{sec:discussion}


Our study highlights that a secure system
requires more than sound abstract mathematical cryptography and UC proofs; a secure system
also needs careful attention to design, engineering, implementation, usability, 
key management, and procedures, and careful attention to many associated details.
Useful, clever, novel cryptography underlies the SecureDNA system~\cite{baum2024CIC}, but
to achieve its security goals, the system also depends on its
protocols, system design, and implementation.

The way the SecureDNA system deals with authenticated channels is problematic and 
inconsistent.  Communications are initiated with one-way TLS, not mTLS.
In addition to TLS, SecureDNA uses its own flawed custom mutual authentication protocol SCEP. 
If SecureDNA trusts TLS, then deploying an additional authentication protocol would
add unnecessary complexity.  
If SecureDNA does not trust TLS, then it should use a strong protocol 
that performs effective cryptographic bindings (see Section~\ref{ssec:recs}).

Our analysis of SecureDNA highlights several important security engineering principles that 
can serve as lessons learned for others.
(1)~Cryptographically bind protocol messages to their full context so that they cannot be misused out of context.
Almost all known structural weaknesses of protocols violate this essential principle.
(2)~Favor proven standard solutions over custom protocols. Carefully document and justify deviations from this principle.
(3)~When designing protocols, work throughout the process with experts in formal-methods analysis to
ensure that appropriate security goals are achieved.
(4)~Be aware that secure channels might fail at the PKI level for a variety of reasons, and where appropriate, 
add a second layer of defense at the application layer.


Although it would be impossible for us to know all of the rationale behind the SecureDNA team's design choices, 
we understand that their decision not to use mTLS was interwoven with their decision
not to use X.509. They explained that they avoided using X.509 because of its complexity,
failure to support threshold roots, and limited support in the Rust ecosystem.
In making this engineering tradeoff, they assumed greater risk in the subtle non-trivial task
of designing their own mutual authentication protocol.

SecureDNA relies heavily on \textit{security operations center (SOC)} software and procedures against rate-limiting attacks and other malicious behaviors. 
Although there is a role for such monitoring, it would be better where possible to strengthen the protocols rather than to rely on SOC as a first line of defense.
Even if auditing procedures are currently meticulously followed, there is no assurance that they will continue to be properly followed in the future. Experience shows that
effective monitoring is expensive and difficult to maintain over long periods
of time. There is no assurance that the SOC software is free of errors, and
all controls must be sensitive to differences among providers.
Although SecureDNA's SOC provides significant protection against our rate-limiting attacks,
better protocol design and engineering practices might reduce the risk of human error in these areas.

The SecureDNA protocols depend strongly on the confidentiality, authentication, and integrity properties of the TLS channel. 
In some settings, however, TLS channels can be defeated due
to poor management of root TLS certificates or
vulnerabilities with a corporate firewall that acts as a MitM TLS proxy to monitor
traffic~\cite{o2016tls}. 
SecureDNA cannot control the security practices of the providers.

The structural vulnerabilities we uncovered in the protocols might have been
avoided if formal-methods analysis had been used throughout the design and evaluation process.  
It would have been helpful to state security goals formally, 
to construct formal protocol models, and 
to perform formal-methods analyses on those models.
It would have been helpful to collaborate with experts who have the capability
to perform such formal-methods analyses.

Open problems include a security review of the software implementation,
including the auditing capabilities.
For example, the SecureDNA system deploys a web form through which {\myB} applies
for an ELT.  Vulnerabilities are very common in such forms, and a vulnerability 
in this form might enable attacks that circumvent {\myB}.
Given that SecureDNA depends critically on TLS, it would be prudent to perform a
thorough security review of its implementation and integration of TLS.

\section{Findings and Recommendations}  
\label{sec:findings-recs}

We summarize our major findings of vulnerabilities and offer recommendations for mitigating them.

\subsection{Findings} 
\label{ssec:findings}

(1)~The custom SCEP protocol achieves only one-way authentication.
This structural weakness enables the adversary 
to circumvent rate limits and mount DoS attacks,
if {\myS} connects with a malicious or corrupted {\myH} or {\myK}.

(2)~The system lacks adequate cryptographic bindings of 
certificates, tokens, and responses, including of responses from {\myH} to {\myS}. 
This undesirable structural weakness prevents the system from detecting
if responses, within a TLS channel, from the hazards database were modified.
If a synthesizer were to reconnect with {\myH} over the same TLS session 
(which the implementation disallows), the
adversary could replay and swap responses from {\myH} without breaking TLS.



(3)~Appendix~\ref{sec:other} points out several additional security issues, including
providers possibly storing passphrases in plain text files, and
depending in part on the assurance of email for auditing alerts and records. 
These examples illustrate difficult challenges stemming from 
engineering tradeoffs, limited options, and lack of control over provider actions.


(4)~From discussions with the SecureDNA team,
we learned that part of SecureDNA's security strategy depends
significantly, not on a model of prevention, but 
on a model of monitor, detect, and respond.
SecureDNA carries out this strategy using non-publicly released source code 
(unavailable to us) and a combination of automatic and manual interventions.  


\subsection{Recommendations} 
\label{ssec:recs}

To mitigate the security vulnerabilities we identified, we recommend
the following:

(1)~To mitigate the vulnerabilities from Section~\ref{sec:binding},
instead of using the custom authentication protocol SCEP, use mTLS~\cite{campbell2020rfc}. 
Given that SecureDNA provides certificate hierarchies, it would be relatively simple
to do so.  Similarly, instead of using a custom certificate structure, start with
a standard one such as X.509~\cite{X509} (see also Appendix~\ref{sec:other}).
In the alternative, as newly released Version 1.1.0 does,
replace SCEP with our suggested SCEP+, as explained in Section~\ref{sec:improvements}.

(2)~In each of the protocols, strengthen the cryptographic binding of messages to context. 
For example, bind each response from {\myH} to the associated query from {\myS}.
Also, bind messages and tokens to the channels 
in which they are used (e.g., see~\cite{golaszewski2024limitations,fidobindingSSR}).

(3)~Perform a thorough security review of SecureDNA's implementation and 
integration of TLS.

(4)~As is true for most systems, it is important to devote great attention
to detailed issues of system design and operations, including
those identified in Appendix~\ref{sec:other}.
For example, SecureDNA does recommend to providers that they store keys
appropriately, including in TPMs~\cite{tpm} when suitable, but 
SecureDNA cannot control such practices.

\section{Conclusion}  
\label{sec:conclusion}

Our main findings are the two structural weaknesses (one-way authentication
and inadequate cryptographic bindings) summarized in Section~\ref{ssec:findings}.
Even if these attacks pose low risks, 
including due to monitoring,
eliminating the underlying structural weaknesses would strengthen the system.
Other attacks might be possible.

SG2 (secrecy of the query) appears solidly protected by the blinding 
mechanisms, and this goal does not depend on the strength of the TLS implementation and integration. The other security goals depend critically on TLS, and the security engineering
devoted to protecting them is less convincing.

We prove that our improved SCEP+ satisfies our precisely-stated
security goals (including mutual authentication) in the DY model.
Release 1.1.0 implements our suggestion.

Our study demonstrates that secure
systems require more than sound abstract mathematical cryptography.
Our study also highlights tensions between security and fielding
a usable system that providers will adopt.
Although a malicious entity could avoid controls by using their own synthesis machine,
or using a synthesis provider that does not screen, 
the SecureDNA system can meaningfully raise the safety of legitimate synthesis labs.
We hope that our analysis and suggested mitigations will strengthen the SecureDNA system.

\section*{Acknowledgments} 

This work builds in part on three student projects at UMBC:
two~\cite{491f2024-1,491f2024-4}
from Sherman's fall 2024 INSuRE cybersecurity research course~\cite{sherman_insure_2017,insure}
and one~\cite{652s2025} from Sherman's cryptology class.
Sherman was supported in part by
the National Science Foundation under
DGE grants 1753681 (SFS) and 2138921 (SaTC).
Sherman, Golaszewski, and Romano were supported in 2024--2025 by
the UMBC cybersecurity exploratory grant program.
Fuchs was supported in 2024--2025 by a UMBC cybersecurity graduate fellowship.
We thank Leonard Foner of the SecureDNA team for helpful discussions,
and we thank Kathleen Romanik for editorial suggestions.



\bibliography{references}

@article{securedna2024,
  author = {
    Baum, C. and Berlips, J. and Chen, W. and Cui, H. and Damg\r{a}rd, I. and Dong, J. and Esvelt, K. M. and Foner, L. and Gao, M. and Gretton, D. and Kysel, M. and Li, J. and Li, X. and Paneth, O. and Rivest, R. L. and Sage-Ling, F. and Shamir, A. and Shen, Y. and Sun, M. and Vaikuntanathan, V. and Van Hauwe, L. and Vogel, T. and Weinstein-Raun, B. and Wang, Y. and Wichs, D. and Wooster, S. and Yao, A. C. and Yu, Y. and Zhang, H. and Zhang, K.
  },
  journal = {arXiv preprint arXiv:2403.14023}, 
  title = {A system capable of verifiably and privately screening global {DNA} synthesis}, 
  year = {2024},
  url = {https://arxiv.org/abs/2403.14023},  
  doi = {arXiv:2403.14023v2}
}

@unpublished{MITArticle,
    author = {Langenkamp, Max and Lin, Andrea and Quach, Alex and Hu, Grace},
    title = {Improving the SecureDNA System},
    year = {2022},
    note = {Course project for 6.857 at MIT},
    url = {https://courses.csail.mit.edu/6.857/2022/projects/Langenkamp-Lin-Quach-Hu.pdf},
}

@misc{sec_dna_github,
    author = {{SecureDNA~Foundation}},
    title = {{SecureDNA} source code},
    year = {2025},
    publisher = {Github},
    howpublished = {\url{https://github.com/SecureDNA}}
}

@misc{secureDNAweb,
    author = {{SecureDNA~Foundation}},
    title = {{SecureDNA} website},
    year = {2025},
    howpublished = {\url{https://securedna.org/}}
}

@misc{baum2020cryptographic,
  title={Cryptographic aspects of {DNA} screening},
  author={Baum, Carsten and Cui, Hongrui and Damg{\aa}rd, Ivan and Esvelt, Kevin and Gao, Mingyu and Gretton, Dana and Paneth, Omer and Rivest, Ron and Vaikuntanathan, Vinod and Wichs, Daniel and others},
  year={2020},
  howpublished = {\url{https://people.csail.mit.edu/rivest/pubs/BE20.pdf}}
}

@article{baum2024CIC,
    author = "Baum, Carsten and Berlips, Jens and Chen, Walther and Damgård, Ivan B. and Esvelt, Kevin M. and Foner, Leonard and Gretton, Dana and Kysel, Martin and Rivest, Ronald L. and Roy, Lawrence and Sage-Ling, Francesca and Shamir, Adi and Vaikuntanathan, Vinod and Hauwe, Lynn Van and Vogel, Theia and Weinstein-Raun, Benjamin and Wichs, Daniel and Wooster, Stephen and Yao, Andrew C. and Yu, Yu",
    journal = "{IACR} {C}ommunications in {C}ryptology",
    publisher = "{I}nternational {A}ssociation for {C}ryptologic {R}esearch",
    title = "Efficient Maliciously Secure   Oblivious Exponentiations",
    volume = "1",
    number = "3",
    date = "2024-10-07",
    year = "2024",
    issn = "3006-5496",
    doi = "10.62056/a66cy7qiu"
}

@article{stinson92,
  title={An explication of secret sharing schemes},
  author={Stinson, Douglas R.},
  journal={Designs, Codes and Cryptography},
  volume={2},
  number={4},
  pages={357--390},
  year={1992},
  publisher={Springer},
note = {10.1007/BF00125203}
}

@inproceedings{herzberg1995proactive,
  title={Proactive secret sharing or: How to cope with perpetual leakage},
  author={Herzberg, Amir and Jarecki, Stanis{\l}aw and Krawczyk, Hugo and Yung, Moti},
  booktitle={Advances in Cryptology—CRYPT0’95: 15th Annual International Cryptology Conference Santa Barbara, California, USA, August 27--31, 1995 Proceedings 15},
  pages={339--352},
  year={1995},
  organization={Springer}
}

@article{shamir1979share,
  title={How to share a secret},
  author={Shamir, Adi},
  journal={Communications of the ACM},
  volume={22},
  number={11},
  pages={612--613},
  year={1979},
  publisher={ACm New York, NY, USA}
}

@manual{cpsamanual,
    author = {Moses D. Liskov and John D. Ramsdell and Joshua D. Guttman and Paul D. Rowe}, 
    title = {The Cryptographic Protocol Shapes Analyzer: A Manual}, 
    year = {2016},
}

@article{rust,
author = {Cartas, Cosmin},
year = {2019},
month = {09},
pages = {45-51},
title = {Rust – {T}he Programming Language for Every Industry},
volume = {19},
journal = {Economy Informatics Journal},
doi = {10.12948/ei2019.01.05}
}

@inproceedings{sbp,
  author       = {Willem Burgers and
                  Roel Verdult and
                  Marko C. J. D. van Eekelen},
  editor       = {Hanne Riis Nielson and
                  Dieter Gollmann},
  title        = {Prevent Session Hijacking by Binding the Session to the Cryptographic
                  Network Credentials},
  booktitle    = {Secure {IT} Systems - 18th Nordic Conference, NordSec 2013, Ilulissat,
                  Greenland, October 18-21, 2013, Proceedings},
  series       = {Lecture Notes in Computer Science},
  volume       = {8208},
  pages        = {33--50},
  publisher    = {Springer},
  year         = {2013},
  url          = {https://doi.org/10.1007/978-3-642-41488-6\_3},
  doi          = {10.1007/978-3-642-41488-6\_3},
  bibsource    = {dblp computer science bibliography, https://dblp.org}
}

@misc{cpsa,
  author = {Guttman, J.D. and Liskov, M.D. and Ramsdell, J.D. and Rowe, P.D.},
  title = {The {C}ryptographic {P}rotocol {S}hapes {A}nalyzer ({CPSA})},
  publisher = {GitHub},
  journal = {GitHub repository},
  howpublished = {\url{https://github.com/mitre/cpsa}}
}

@article{dolevyao,
    author = {Dolev, D. and Yao, A.},
    journal = {IEEE Transactions on Information Theory}, 
    title = {On the security of public key protocols}, 
    year = {1983},
    volume = {29},
    number = {2},
    pages = {198-208},
    doi = {10.1109/TIT.1983.1056650}
}

@article{Lowe1995,
  author       = {Gavin Lowe},
  title        = {An Attack on the {N}eedham-{S}chroeder {P}ublic-{K}ey {A}uthentication {P}rotocol},
  journal      = {Inf. Process. Lett.},
  volume       = {56},
  number       = {3},
  pages        = {131--133},
  year         = {1995},
  url          = {https://doi.org/10.1016/0020-0190(95)00144-2},
  doi          = {10.1016/0020-0190(95)00144-2},
  bibsource    = {dblp computer science bibliography, https://dblp.org}
}

@inproceedings{golaszewski2024limitations,
  title={Limitations of Wrapping Protocols and {TLS} Channel Bindings: Formal-Methods Analysis of the {S}ession {B}inding {P}roxy Protocol},
  author={Golaszewski, Enis and Zieglar, Edward and Sherman, Alan T and Abou Elsaad, Kirellos and Fuchs, Jonathan D},
  booktitle={International Conference on Research in Security Standardisation ({SSR} 2024)},
  pages={81--119},
  year={2024},
  organization={Springer}
}

@inproceedings{fidobindingSSR,
    title = {Cryptographic Binding Should Not Be Optional: A Formal-Methods Analysis of {FIDO UAF} Authentication},
    author = {Enis Golaszewski and Alan T. Sherman and Edward Zieglar},
    year = {2025},
    month = {December},
    booktitle    = {International Conference on Research in Security Standardisation ({SSR} 2025)},
    series       = {Lecture Notes in Computer Science},
    publisher    = {Springer},
    note  = {in press},
}

@inproceedings{blanchet2016modeling,
  author       = {Bruno Blanchet},
  editor       = {Alessandro Aldini and
                  Javier L{\'{o}}pez and
                  Fabio Martinelli},
  title        = {Automatic Verification of Security Protocols in the Symbolic Model:
                  The Verifier {P}ro{V}erif},
  booktitle    = {Foundations of Security Analysis and Design {VII} - {FOSAD} 2012/2013
                  Tutorial Lectures},
  series       = {Lecture Notes in Computer Science},
  volume       = {8604},
  pages        = {54--87},
  publisher    = {Springer},
  year         = {2013},
  url          = {https://doi.org/10.1007/978-3-319-10082-1\_3},
  doi          = {10.1007/978-3-319-10082-1\_3},
  bibsource    = {dblp computer science bibliography, https://dblp.org}
}

@inproceedings{meier2013tamarin,
  author       = {Simon Meier and
                  Benedikt Schmidt and
                  Cas Cremers and
                  David A. Basin},
  editor       = {Natasha Sharygina and
                  Helmut Veith},
  title        = {The {TAMARIN} Prover for the Symbolic Analysis of Security Protocols},
  booktitle    = {Computer Aided Verification - 25th International Conference, {CAV}
                  2013, Saint Petersburg, Russia, July 13-19, 2013. Proceedings},
  series       = {Lecture Notes in Computer Science},
  volume       = {8044},
  pages        = {696--701},
  publisher    = {Springer},
  year         = {2013},
  url          = {https://doi.org/10.1007/978-3-642-39799-8\_48},
  doi          = {10.1007/978-3-642-39799-8\_48},
  bibsource    = {dblp computer science bibliography, https://dblp.org}
}

@inproceedings{escobar2009maude,
  author       = {Santiago Escobar and
                  Catherine Meadows and
                  Jos{\'{e}} Meseguer},
  editor       = {Alessandro Aldini and
                  Gilles Barthe and
                  Roberto Gorrieri},
  title        = {{M}aude-{NPA}: Cryptographic Protocol Analysis Modulo Equational Properties},
  booktitle    = {Foundations of Security Analysis and Design V, {FOSAD} 2007/2008/2009
                  Tutorial Lectures},
  series       = {Lecture Notes in Computer Science},
  volume       = {5705},
  pages        = {1--50},
  publisher    = {Springer},
  year         = {2007},
  url          = {https://doi.org/10.1007/978-3-642-03829-7\_1},
  doi          = {10.1007/978-3-642-03829-7\_1},
  bibsource    = {dblp computer science bibliography, https://dblp.org}
}

@article{canetti2020universally,
  title={Universally composable security},
  author={Canetti, Ran},
  journal={Journal of the ACM (JACM)},
  volume={67},
  number={5},
  pages={1--94},
  year={2020},
  publisher={ACM New York, NY, USA}
}

@article{fabrega1999strand,
  author       = {F. Javier Thayer and
                  Jonathan C. Herzog and
                  Joshua D. Guttman},
  title        = {Strand Spaces: Proving Security Protocols Correct},
  journal      = {J. Comput. Secur.},
  volume       = {7},
  number       = {1},
  pages        = {191--230},
  year         = {1999},
  url          = {https://doi.org/10.3233/jcs-1999-72-304},
  doi          = {10.3233/JCS-1999-72-304},
  bibsource    = {dblp computer science bibliography, https://dblp.org}
}

@techreport{liskov2011completeness,
  title={Completeness of {CPSA}},
  author={Liskov, Moses and Rowe, Paul and Thayer, Javier},
  year={2011},
  institution={MITRE},
  note = {\url{https://www.mitre.org/sites/default/files/pdf/12_0038.pdf}}
}

@manual{liskov2016cryptographic,
    title = {The Cryptographic Protocol Shapes Analyzer: A Manual},
    author = {Liskov, Moses D and Ramsdell, John D and Guttman, Joshua D and Rowe, Paul D},
    organization = {The MITRE Corporation},
    year = {2016}
}

@article{Needham1978,
     author = {Needham, Roger M. and Schroeder, Michael D.},
     title = {{Using encryption for authentication in large networks of computers}},
     journal = {Commun. ACM},
     issue_date = {Dec. 1978},
     volume = {21},
     number = {12},
     month = dec,
     year = {1978},
     issn = {0001-0782},
     pages = {993--999},
     numpages = {7},
     url = "http://doi.acm.org/10.1145/359657.359659",
     doi = {10.1145/359657.359659},
     acmid = {359659},
     publisher = {ACM},
     address = {New York, NY, USA},
}

@book{rescorla2001ssl,
  title={SSL and TLS: Designing and Building Secure Systems},
  author={Rescorla, Eric},
  publisher = {Addison-Wesley},
  year={2001}
}

@article{rowe2016measuring,
  author       = {Paul D. Rowe and
                  Joshua D. Guttman and
                  Moses D. Liskov},
  title        = {Measuring protocol strength with security goals},
  journal      = {Int. J. Inf. Sec.},
  volume       = {15},
  number       = {6},
  pages        = {575--596},
  year         = {2016},
  url          = {https://doi.org/10.1007/s10207-016-0319-z},
  doi          = {10.1007/S10207-016-0319-Z},
  bibsource    = {dblp computer science bibliography, https://dblp.org}
}

@InProceedings{lowe_authentication,
  author =       {Gavin Lowe},
  title =        {A Hierarchy of Authentication Specifications},
  booktitle =    {10th Computer Security Foundations
                  Workshop Proceedings},
  year =         1997,
  publisher =    {{IEEE} CS Press},
  pages =        {31--43}
}

@misc{rfc8446,
    series =    {Request for Comments},
    number =    8446,
    howpublished =  {RFC 8446},
    publisher = {RFC Editor},
    doi =       {10.17487/RFC8446},
    url =       {https://www.rfc-editor.org/info/rfc8446},
    author =    {Eric Rescorla},
    title =     {{The Transport Layer Security (TLS) Protocol Version 1.3}},
    pagetotal = 160,
    year =      2018,
    month =     aug,
}

@misc{rfc5246,
    series =    {Request for Comments},
    number =    5246,
    howpublished =  {RFC 5246},
    publisher = {RFC Editor},
    doi =       {10.17487/RFC5246},
    url =       {https://www.rfc-editor.org/info/rfc5246},
    author =    {Eric Rescorla and Tim Dierks},
    title =     {{The Transport Layer Security (TLS) Protocol Version 1.2}},
    pagetotal = 104,
    year =      2008,
    month =     aug,
}

@article{sherman_insure_2017,
	title        = {The {INSuRE} {Project}: {CAE}-{Rs} collaborate to engage students in cybersecurity research},
	author       = {Sherman, Alan T. and Dark, M. and Chan, A. and Morris, T. and Oliva, L. and Springer, J. and Thuraisingham, B. and Vatcher, C. and Verma, R. and Wetzel, S.},
	year         = 2017,
	month        = aug,
	journal      = {IEEE Security \& Privacy},
	volume       = 15,
	number       = 4
}

@misc{insure,
    title        = {{Information Security Research and Education (INSuRE)}},
    howpublished = {\url{https://caecommunity.org/initiative/insure}},
    note = "[Online; accessed 15-May-2025]"
}

@misc{X509,
    title = {X.509 : Information technology - Open Systems Interconnection - The Directory: Public-key and 
    attribute certificate frameworks Recommendation {X}.509},
    howpublished = {\url{https://www.itu.int/rec/T-REC-X.509}}
}

@inproceedings{boneh1998decision,
  title={The decision {D}iffie-{H}ellman problem},
  author={Boneh, Dan},
  booktitle={International algorithmic number theory symposium},
  pages={48--63},
  year={1998},
  organization={Springer}
}

@misc{campbell2020rfc,
  title={{RFC} 8705: OAuth 2.0 Mutual-TLS Client Authentication and Certificate-Bound Access Tokens},
  author={Campbell, B and Bradley, J and Sakimura, N and Lodderstedt, T},
  year={2020},
  publisher={RFC Editor}
}

@article{Thompson2018,
  title = {Student misconceptions about cybersecurity concepts: Analysis of think-aloud interviews},
  author = {Thompson, Julia D. and Herman, Geoffrey L. and Scheponik, Travis and Oliva, Linda and Sherman, Alan T. and Golaszewski, Ennis and Phatak, Dhananjay and Patsourakos, Kostantinos},
  journal = {Journal of Cybersecurity Education, Research and Practice},
  volume = {2018},
  number = {1},
  pages = {5},
  year = {2018}
}

@article{barrangou2016applications,
  title={Applications of {CRISPR} technologies in research and beyond},
  author={Barrangou, Rodolphe and Doudna, Jennifer A},
  journal={Nature biotechnology},
  volume={34},
  number={9},
  pages={933--941},
  year={2016},
  publisher={Nature Publishing Group US New York}
}

@misc{tpm,
    author = {{Trusted Computing Group (TCG)}},
    title = {{TPM 2.0 Library}},
    year = {2025},
    howpublished = {\url{https://trustedcomputinggroup.org/resource/tpm-library-specification/}}
}

@techreport{malpani2013x,
  title={X.509 {I}nternet public key infrastructure online certificate status protocol-ocsp},
  author={Malpani, A and Galperin, S and Adams, C},
  year={2013},
  institution={IETF, 2013. RFC 6960}
}

@article{bernstein2017post,
  title={Post-quantum cryptography},
  author={Bernstein, Daniel J and Lange, Tanja},
  journal={Nature},
  volume={549},
  number={7671},
  pages={188--194},
  year={2017},
  publisher={Nature Publishing Group UK London}
}

@article{HARRIS1999885,
title = {{TCP/IP} security threats and attack methods},
journal = {Computer Communications},
volume = {22},
number = {10},
pages = {885-897},
year = {1999},
issn = {0140-3664},
doi = {https://doi.org/10.1016/S0140-3664(99)00064-X},
url = {https://www.sciencedirect.com/science/article/pii/S014036649900064X},
author = {B. Harris and R. Hunt}
}

@misc{652s2025,
    author = {Jeremy Romanik Romano},
    title = {Analysis of {SecureDNA} System Key Management},
    howpublished = {CMSC-652 Cryptology, University of Maryland, Baltimore County},
    year = {2025},
    month = {May}
}

@misc{491f2024-1,
    author = {Jeremy Romanik Romano and Sairam Bokka and Zach Heck and
Tobbi Caplan and William Zheng and Sean Stultz},
    title = {Modeling and Formal Analysis of Authentication in the {SecureDNA} Protocol},
    howpublished = {CMSC-491/691: Cybersecurity Research (INSuRE), University of Maryland, Baltimore County},
    year = {2024},
    month = {December}
}

@misc{491f2024-4,
    author = {Ben Guest and Isaiah Roland-Reid and Tejas Ramlal Walke and Maitreyee Asatkar and Tanmay Sarkar},
    title = {Security Analysis of Registration, Exemption Handling, and Web App Security of the {SecureDNA} Protocol},
    howpublished = {CMSC-491/691: Cybersecurity Research (INSuRE), University of Maryland, Baltimore County},
    year = {2024},
    month = {December}
}

@inproceedings{o2016tls,
  title={TLS proxies: Friend or foe?},
  author={O'Neill, Mark and Ruoti, Scott and Seamons, Kent and Zappala, Daniel},
  booktitle={Proceedings of the 2016 Internet Measurement Conference},
  pages={551--557},
  year={2016}
}

@misc{AI2027,
    author = {Daniel Kokotajlo and Scott Alexander and Thomas Larsen and Eli Liflandand Romeo Dean},
    title = {{AI} 2027},
    howpublished = {\url{https://ai-2027.com/}},
    year = {2025},
    month = {April}
}

@misc{anonGitHub,
    author = {Anonymous},
    title = {{SecureDNA CPSA Model Repository}},
    month = {July},
    year = {2025},
    publisher = {GitHub},
    journal = {GitHub repository},
    howpublished = {https://anonymous.4open.science/r/securedna-cpsa-anon-BB46/README.md}}

@article{kane2024screening,
  title={Screening state of play: the biosecurity practices of synthetic DNA providers},
  author={Kane, Arianne and Parker, Michael T},
  journal={Applied Biosafety},
  volume={29},
  number={2},
  pages={85--95},
  year={2024},
  publisher={Mary Ann Liebert, Inc., publishers 140 Huguenot Street, 3rd Floor New~…}
}

@article{hoffmann2023safety,
  title={Safety by design: Biosafety and biosecurity in the age of synthetic genomics},
  author={Hoffmann, Stefan A and Diggans, James and Densmore, Douglas and Dai, Junbiao and Knight, Tom and Leproust, Emily and Boeke, Jef D and Wheeler, Nicole and Cai, Yizhi},
  journal={Iscience},
  volume={26},
  number={3},
  year={2023},
  publisher={Elsevier}
}

@misc{JHUhub,
Author = {Johns Hopkins Bloomberg School of Public Health Center for Health Science},
Title = {Gene Synthesis Screening Information Hub},
Howpublished = {\url{https://genesynthesisscreening.centerforhealthsecurity.org/}},
}

@misc{NIHguidelines,
Author = {National Science Advisory Board for Biosecurity},
Title = {Addressing Biosecurity Concerns Related to the Synthesis of Select Agents},
Howpublished = {National Institutes
of Health, Bethesda, MD},
Year = {2006},
}

@misc{IGSCprotocol,
Author = {International Gene Synthesis Consortium},
Title = {Harmonized Screening Protocol v2.0},
Howpublished = {\url{ https://genesynthesisconsortium.
org/wp-content/uploads/IGSCHarmonizedProtocol11-21-17.pdf}},
year = {2017},
}
\bibliographystyle{ieeetr}


\appendix

\subsection{Ethical Considerations}

We have responsibly disclosed our findings to the SecureDNA team by sharing preliminary drafts of our paper, exchanging emails, and meeting with them remotely. 
These meetings took place on June 3, 2025, and July 30, 2025.  
Newly released Version 1.1.0 of the SecureDNA system fixes SCEP with our proposed SCEP+ protocol.
It also implements an optional ``verifiable screening'' mode, which can mitigate the response-switching attack.



\clearpage
\subsection{Formal-Methods Analysis of Basic Query}
\label{append:basic}

Building on the analysis of SCEP, we define a strand space for the SecureDNA basic query protocol, formalize a subset of the protocol's informal security goals (Section~\ref{ssec:goals}), and prove theorems that incorporate the formal goals.


\subsubsection{Query Protocol Strand Space Model} 
\label{ssec:Q-models}

As we did for the SCEP protocol, we formalize the behavior of the DY adversary and honest communicants as penetrator and regular strands.
Definition~\ref{def:pen-strands} defines the set of penetrator strands.
We formalize a subset of the query protocol security goals as logical formulae on execution models in the query strand space.

To define regular strands that carry out the roles of  {\myS},  {\myK}, and {\myH}, we first define traces for each of these roles.
We define two complementary traces alongside 
the {\myS} trace (Definition~\ref{def:query-s-trace}): {\myK} trace (Definition~\ref{def:query-k-trace}), and keyed-hash {\myH} trace (Definition~\ref{def:query-h-trace}).
These regular traces incorporate traces of SCEP (Definition~\ref{def:scep-traces}).


\begin{definition}
\label{def:query-terms}
[Query Protocol Terms]
Let $\vec{Q}$ be a list of terms:
\begin{enumerate}
    \item $\mathcal{S}, \mathcal{K}, \mathcal{H}, M, CA \in Name$
    \item $r_\mathcal{S}, r'_\mathcal{S}, r''_\mathcal{S},r'''_\mathcal{S} \in Nonce$
    \item $r_\mathcal{K}, r_\mathcal{K}', r_\mathcal{S}, r_\mathcal{S}' \in Nonce$
    \item $s, \beta, k, e_\mathcal{S}, e_\mathcal{S}', e_\mathcal{K}, e_\mathcal{H} \in Expt$
    \item $M(s): g^s$
    \item $\mathcal{SK}_{WRITE}: h(g^{e_\mathcal{S} \times e_\mathcal{K}}||r_\mathcal{S}'||r_\mathcal{K}', \text{``client-write"})$
    \item $\mathcal{SH}_{WRITE}: h(g^{e_\mathcal{S}' \times e_\mathcal{H}}||r_\mathcal{S}'''||r_\mathcal{H}', \text{``client-write"})$
    \item $K_{WRITE}: h(g^{e_\mathcal{S} \times e_\mathcal{K}}||r_\mathcal{S}'||r_\mathcal{K}', \text{``server-write"})$
    \item $H_{WRITE}: h(g^{e_\mathcal{S}' \times e_\mathcal{H}}||r_\mathcal{S}'''||r_\mathcal{H}', \text{``server-write"})$.
\end{enumerate}
\end{definition}

\begin{definition}[Query Protocol Synthesizer Trace]
\label{def:query-s-trace}

Let trace $Tr_{QUERY-\mathcal{S}}(\vec{Q})$ comprise the following sequence of events:
\begin{enumerate}
    \item $Tr_{SCEP-\mathcal{SR}}(\\\mathcal{S}, \mathcal{K}, M, CA, r_\mathcal{S}, r_\mathcal{K},r_\mathcal{S}', r_\mathcal{K}', \omega_\mathcal{K}, e_\mathcal{S}, e_\mathcal{K})$
    \item $+e(\omega_\mathcal{K}|| M(s)^\beta,\mathcal{SK}_{WRITE})$
    \item $-e(M(s)^{\beta k}, \mathcal{K}_{WRITE})$
    \item $Tr_{SCEP-\mathcal{SR}}(\\\mathcal{S}, \mathcal{H}, M, CA, r_\mathcal{S}'', r_\mathcal{H}, r_\mathcal{S}''', r_\mathcal{H}', \omega_\mathcal{H}, e_\mathcal{S}', e_\mathcal{H})$
    \item $+e(\omega_\mathcal{H}, M(s)^k, \mathcal{SH}_{WRITE})$
    \item $-e(Resp, \mathcal{H}_{WRITE})$.
\end{enumerate}

\end{definition}

\begin{definition}
[Query Protocol Keyserver Trace]
\label{def:query-k-trace}

Let trace $Tr_{QUERY-\mathcal{K}}(\vec{Q})$ comprise the following sequence of events:
\begin{enumerate}
    \item $Tr_{SCEP-\mathcal{RS}}(\\\mathcal{S}, \mathcal{K}, M, CA, r_\mathcal{S}, r_\mathcal{K},r_\mathcal{S}', r_\mathcal{K}', \omega_\mathcal{K}, e_\mathcal{S}, e_\mathcal{K})$
    \item $-e(\omega_\mathcal{K}|| M(s)^\beta,\mathcal{SK}_{WRITE})$
    \item $+e(M(s)^{\beta k}, \mathcal{K}_{WRITE})$.
\end{enumerate}

\end{definition}

\begin{definition}
[Query Protocol Keyed-Hash Database Trace]
\label{def:query-h-trace}
Let trace $Tr_{QUERY-\mathcal{H}}(\vec{Q})$ comprise the following sequence of events:
\begin{enumerate}
    \item $Tr_{SCEP-\mathcal{RS}}(\\\mathcal{S}, \mathcal{H}, M, CA, r_\mathcal{S}'', r_\mathcal{H}, r_\mathcal{S}''', r_\mathcal{H}', \omega_\mathcal{H}, e_\mathcal{S}', e_\mathcal{H})$
    \item $-e(\omega_\mathcal{H}, M(s)^k, \mathcal{SH}_{WRITE})$
    \item $+e(Resp, \mathcal{H}_{WRITE})$.
\end{enumerate}

\end{definition}

To define the query protocol strand space (Definition~\ref{def:query-strandspace}), we specify a union of the regular query protocol strands (Definition~\ref{def:query-protocol-strands}) and the penetrator strands.

\begin{definition}
[Query Protocol Strands]
\label{def:query-protocol-strands}
We define three sets of regular strands that carry out the query protocol.
\begin{enumerate}
    \item Let $\Sigma_{QUERY-\mathcal{S}}$ be the set of all strands with the trace of the form $Tr_{QUERY-\mathcal{S}}$.
    \item Let $\Sigma_{QUERY-\mathcal{K}}$ be the set of all strands with the trace of the form $Tr_{QUERY-\mathcal{K}}$.
    \item Let $\Sigma_{QUERY-\mathcal{H}}$ be the set of all strands with the trace $Tr_{QUERY-\mathcal{H}}$.
\end{enumerate}
\end{definition}

\begin{definition}[Query Protocol Strand Space]
\label{def:query-strandspace}
The query protocol strand space comprises $\Sigma_{pen} \cup \Sigma_{QUERY-\mathcal{S}} \cup \Sigma_{QUERY-\mathcal{K}} \cup \Sigma_{QUERY-\mathcal{H}}$.
\end{definition}

\subsubsection{Security Goals} 
\label{ssec:Q-goals}

Using our formal definitions of confidentiality and agreement (Section~\ref{ssec:scep-security-goals}), we formalize a subset (SG2, SG3) of the informal security goals from Section~\ref{ssec:goals}.
As we do in our SCEP analysis, we assume that the adversary does not possess initial knowledge of honest communicant secret keys or TLS DH exponents, including the CA secret key.
Additionally, each security goal makes freshness assumptions appropriate to the cryptographic perspective of the goal.
In our analysis, we combine the roles {\myS} and {\myC} into the role {\myS} because communication between these entities falls outside of our scope.
The informal security goals SG2 and SG3 hold iff the formal security goals hold for all possible execution models under the corresponding formal assumptions.

Because the informal SG2 asserts that no party other than {\myS} and {\myC} learn the synthesis order request, we formalize SG2 as a confidentiality goal on the sequence $M(s)$.
We formulate formal SG2 security goals from the perspective of each formal protocol role: {\myS} (Definition~\ref{def:SG2-synthclient}), {\myK} (Definition~\ref{def:SG2-keyserver}), and {\myH} (Definition~\ref{def:SG2-database}).
Together, these definitions assert that, for a given role trace $Tr$, a strand with 
trace $Tr$ under the origination assumptions corresponding to the role, 
does not leak $M(s)$ to the adversary.
Should the conjunction of the three definitions holds, we prove that $M(s)$ cannot leak to the adversary in any execution model resulting from the query strand space.

\begin{definition}[SG2: Synthesizer Perspective]
\label{def:SG2-synthclient}
Let \textbf{orig} be the conjunction of the following origination assumptions on $\vec{Q}$:
\begin{enumerate}
    \item $Non(sk(\mathcal{S}))$, $Non(sk(\mathcal{K}))$, $Non(sk(\mathcal{H}))$, $Non(sk(M))$, $Non(sk(CA))$, where $Non$ indicates non-origination, and
    \item $Uniq(\omega_\mathcal{K})$, $Uniq(\omega_\mathcal{H})$, $Uniq(r_\mathcal{S})$, $Uniq(r'_\mathcal{S})$, $Uniq(r''_\mathcal{S})$, $Uniq(r'''_\mathcal{S})$, $Uniq(e_\mathcal{S})$, $Uniq(e_\mathcal{K}))$, $Uniq(e_\mathcal{H}))$, $Uniq(s)$, $Uniq(\beta)$ where $Uniq$ indicates unique origination.
\end{enumerate}
A query protocol strand space $\mathcal{P}$ satisfies SG2 from the synthesizer perspective iff $Conf_\mathcal{P}(Tr_{QUERY-\mathcal{S}}(\vec{Q}), M(s), \mathbf{orig}$) is true.
\end{definition}

\begin{definition}[SG2: Keyserver Perspective]
\label{def:SG2-keyserver}
Let \textbf{orig} be the conjunction of the following origination assumptions on $\vec{Q}$:
\begin{enumerate}
    \item $Non(sk(\mathcal{S}))$, $Non(sk(\mathcal{K}))$, $Non(sk(\mathcal{H}))$, $Non(sk(M))$, $Non(sk(CA))$, where $Non$ indicates non-origination, and
    \item $Uniq(\omega_\mathcal{K})$, $Uniq(\omega_\mathcal{H})$, $Uniq(r_\mathcal{K})$, $Uniq(r'_\mathcal{K})$, $Uniq(e_\mathcal{S})$, $Uniq(e_\mathcal{K}))$, $Uniq(e_\mathcal{H}))$, $Uniq(s)$, $Uniq(\beta)$ where $Uniq$ indicates unique origination.
\end{enumerate}
A query protocol strand space $\mathcal{P}$ satisfies SG2 from the keyserver perspective iff $Conf_\mathcal{P}(Tr_{QUERY-\mathcal{K}}(\vec{Q}), M(s), \mathbf{orig}$) is true.
\end{definition}

\begin{definition}[SG2: Keyed-Hash Database Perspective]
\label{def:SG2-database}
Let \textbf{orig} be the conjunction of the following origination assumptions on $\vec{Q}$:
\begin{enumerate}
    \item $Non(sk(\mathcal{S}))$, $Non(sk(\mathcal{K}))$, $Non(sk(\mathcal{H}))$, $Non(sk(M))$, $Non(sk(CA))$, where $Non$ indicates non-origination, and
    \item $Uniq(\omega_\mathcal{K})$, $Uniq(\omega_\mathcal{H})$, $Uniq(r_\mathcal{H})$, $Uniq(r'_\mathcal{H})$, $Uniq(e_\mathcal{S})$, $Uniq(e_\mathcal{H}))$, $Uniq(e_\mathcal{H}))$, $Uniq(s)$, $Uniq(\beta)$ where $Uniq$ indicates unique origination.
\end{enumerate}
A query protocol strand space $\mathcal{P}$ satisfies SG2 from the keyserver perspective iff $Conf_\mathcal{P}(Tr_{QUERY-\mathcal{H}}(\vec{Q}), M(s), \mathbf{orig}$) is true.
\end{definition}

Next we formalize SG3. Informally, SG3 asserts that {\myS} receives a valid answer as to whether a synthesis order request is in the database.
We formalize SG3 as a pair of agreement goals in which (1)~strands of role {\myS} must agree on the encrypted sequence $M(s)^k$ and the query response $Resp$ with a regular strand of role {\myH}, and (2)~strands of role {\myH} must agree on $M(s)^k, Resp$ with a regular strand of {\myS}.
If this pair of goals (Definition~\ref{def:SG3-synthclient}, Definition~\ref{def:SG3-database}) holds true for all execution models in our strand space, 
we prove {injective agreement}~\cite{lowe_authentication}
for $M(s)^k, Resp$ between {\myS} and {\myH}.
Because $Resp$ must originate on a regular {\myH} strand that faithfully follows the protocol, informal SG3 security goal holds iff
there is injective agreement on $M(s)^k, Resp$.

\begin{definition}[SG3: Synthclient Perspective]
\label{def:SG3-synthclient}
Let \textbf{orig} be the conjunction of the following origination assumptions on $\vec{Q}$:
\begin{enumerate}
    \item $Non(sk(\mathcal{S}))$, $Non(sk(\mathcal{K}))$, $Non(sk(\mathcal{H}))$, $Non(sk(M))$, $Non(sk(CA))$, where $Non$ indicates non-origination, and
    \item $Uniq(\omega_\mathcal{K})$, $Uniq(\omega_\mathcal{H})$, $Uniq(r_\mathcal{S})$, $Uniq(r'_\mathcal{S})$, $Uniq(r''_\mathcal{S})$, $Uniq(r'''_\mathcal{S})$, $Uniq(e_\mathcal{S})$, $Uniq(e_\mathcal{K}))$, $Uniq(e_\mathcal{H}))$, $Uniq(s)$, $Uniq(\beta)$ where $Uniq$ indicates unique origination.
\end{enumerate}
A query protocol strand space $\mathcal{P}$ satisfies SG3 from the synthesizer perspective iff $Agree_\mathcal{P}(Tr_{QUERY-\mathcal{S}}(\vec{Q}), Tr_{QUERY-\mathcal{H}}(\vec{Q})$,\linebreak $[M(s)^k,Resp], \mathbf{orig})$ is true.
\end{definition}

\begin{definition}[SG3: Keyed-Hash Database Perspective]
\label{def:SG3-database}
Let \textbf{orig} be the conjunction of the following origination assumptions on $\vec{Q}$:
\begin{enumerate}
    \item $Non(sk(\mathcal{S}))$, $Non(sk(\mathcal{K}))$, $Non(sk(\mathcal{H}))$, $Non(sk(M))$, $Non(sk(CA))$, where $Non$ indicates non-origination, and
    \item $Uniq(\omega_\mathcal{K})$, $Uniq(\omega_\mathcal{H})$, $Uniq(r_\mathcal{S})$, $Uniq(r'_\mathcal{S})$, $Uniq(r''_\mathcal{S})$, $Uniq(r'''_\mathcal{S})$, $Uniq(e_\mathcal{S})$, $Uniq(e_\mathcal{K}))$, $Uniq(e_\mathcal{H}))$, $Uniq(s)$, $Uniq(\beta)$ where $Uniq$ indicates unique origination.
\end{enumerate}
A query protocol strand space $\mathcal{P}$ satisfies SG3 from the database perspective iff $Agree_\mathcal{P}(Tr_{QUERY-\mathcal{H}}(\vec{Q}), Tr_{QUERY-\mathcal{S}}(\vec{Q})$, \linebreak $[M(s)^k,Resp], \mathbf{orig})$ is true.
\end{definition}

\subsubsection{Analysis} 
\label{ssec:Q-analysis}

We use CPSA to describe all unique shapes possible within a strand space under the origination assumptions of each query protocol security goal.
Theorem~\ref{thrm:query-conf} asserts the confidentiality of the client sequence $M(s)$.
Theorem~\ref{thrm:query-agree} asserts shows that {\myH} fails to agree on the critical values $M(s)$ and the query response $Resp$ with a legitimate synthesizer {\myS}.
This failure results from the SCEP's failure to mutually authenticate (Section~\ref{ssec:scep-analysis}).

\begin{theorem}
[Confidentiality of $M(s)$]
\label{thrm:query-conf}
If $\mathcal{P}$ is any query protocol strand space 
(Definition~\ref{def:query-strandspace}),
then $\mathcal{P}$ satisfies SG2 from each role's perspective. (Definition~\ref{def:SG2-synthclient}, Definition~\ref{def:SG2-keyserver},
Definition~\ref{def:SG2-database})
\end{theorem}

\begin{proof}[Proof by Enumeration]
After successfully terminating, CPSA identifies all unique shapes corresponding to each security goal.
For each security goal, there exists no contradicting shape.
Thus, $M(s)$ remains confidential from each perspective.
\end{proof}

\begin{theorem}
[Agreement on $M(s)^k, Resp$]
\label{thrm:query-agree}
If $\mathcal{P}$ is any query protocol strand space 
(Definition~\ref{def:query-strandspace}),
then $\mathcal{P}$ does \textbf{not} satisfy SG3 from the keyed-hash database perspectives.
(Definition~\ref{def:SG3-database})
\end{theorem}

\begin{proof}[Proof by Counterexample]
After terminating, CPSA enumerates several unique shapes under the assumptions of Definition~\ref{def:SG3-database}, each of which contradict the security goal (no agreement on $M(s), Resp$) because {\myH} is unable to authenticate a legitimate instance of {\myS}.
\end{proof}

\clearpage 
\subsection{Formal-Methods Analysis of Exemption Query}
\label{append:exemption}

The SecureDNA exemption-handling protocol builds on the query protocol by adding the exemption-list token and an authenticator device.
The exemption protocol replicates the weakness of the query protocol because it also relies on the flawed SCEP to authenticate {\myS} to infrastructure servers {\myK} and {\myH}.


\subsubsection{Models} 
\label{ssec:E-models}

To define the exemption protocol traces, we build on existing query protocol traces by incorporating additional terms of the exemption protocol.
First, we extend the query protocol terms (Definition~\ref{def:query-terms}) to include the exemption protocol terms.
Second, we extend the query protocol traces to include these new terms.
Third, we extend the query protocol security goals (Appenidx~\ref{ssec:E-goals}).
Finally, we define and prove security theorems for the exemption protocol (Appendix~\ref{ssec:E-analysis}).

To extend the query protocol terms, we define an exemption sequence $M(s_E)$, where $s_E$ is the client's exemption sequence, and $AuthCode$ is a nonce that models an assertion by the authentication backend {\myA}.
Because we do not analyze the DOPRF, we model the exemption sequence as a plain DH exponentiation, as we do for the client sequence $M(s)$.
Additionally, we introduce names for {\myB} that back the exemption token $T_{\mathcal{C, B}}^E$ of  {\myC}.
The synthesizer includes $T_{\mathcal{C, B}}^E$ when making requests to {\myH}.

\begin{definition}
[Exemption Protocol Terms]
\label{def:query-terms2}
Let $\vec{E}$ be a list of terms that includes the terms of $\vec{Q}$ (Definition~\ref{def:query-terms}) and the terms below:
\begin{enumerate}
    \item $\mathcal{B, C} \in Name$
    \item $s_E \in Expt$
    \item $M(s_E): g^{s_E}$
    \item $AuthCode \in Nonce$.
\end{enumerate}
\end{definition}

Next, we define traces for the exemption protocol.
In our definitions, we specify messages only that are \textit{distinct} from those in the basic query protocol.

\begin{definition}[Exemption Protocol Synthesizer Trace]
\label{def:exempt-s-trace}

Let trace $Tr_{EXEMPT-\mathcal{S}}(\vec{E})$ comprise the following sequence of events:
\begin{enumerate}
    \item (Definition~\ref{def:query-s-trace})
    \item $+e(\omega_\mathcal{K}|| M(s)^\beta||M(s_E)^\beta,\ \mathcal{SK}_{WRITE})$
    \item $-e(M(s)^{\beta k}||M(s_E)^{\beta k}, \ \mathcal{K}_{WRITE})$
    \item (Definition~\ref{def:query-s-trace})
    \item $+e(\omega_\mathcal{H}, M(s)^k||T^E_{\mathcal{C, B}}||AuthCode||M(s_E)^k, \\ \mathcal{SH}_{WRITE})$
    \item (Definition~\ref{def:query-s-trace}).
\end{enumerate}

\end{definition}

\begin{definition}
[Exemption Protocol Keyserver Trace]
\label{def:exempt-k-trace}

Let trace $Tr_{EXEMPT-\mathcal{K}}(\vec{E})$ comprise the following sequence of events:
\begin{enumerate}
    \item (Definition~\ref{def:query-k-trace})
    \item $-e(\omega_\mathcal{K}|| M(s)^\beta||M(s_E)^\beta,\mathcal{SK}_{WRITE})$
    \item $+e(M(s)^{\beta k}||M(s_E)^{\beta k}, \ \mathcal{K}_{WRITE})$.
\end{enumerate}

\end{definition}

\begin{definition}
[Exemption Protocol Keyed-Hash Database Trace]
\label{def:exempt-h-trace}
Let trace $Tr_{EXEMPT-\mathcal{H}}(\vec{E})$ comprise the following sequence of events:
\begin{enumerate}
    \item (Definition~\ref{def:query-h-trace})
    \item $-e(\omega_\mathcal{H}, M(s)^k||T^E_{\mathcal{C, B}}||AuthCode||M(s_E)^k, \\ \mathcal{SH}_{WRITE})$
    \item (Definition~\ref{def:query-h-trace}).
\end{enumerate}
\end{definition}

\begin{definition}[Exemption Protocol Strand Space]
\label{def:query-strandspaceX}
The exemption protocol strand space comprises $\Sigma_{pen} \cup \Sigma_{EXEMPT-\mathcal{S}} \cup \Sigma_{EXEMPT-\mathcal{K}} \cup \Sigma_{EXEMPT-\mathcal{H}}$,
where $\Sigma_{EXEMPT-\mathcal{S}}$, $\Sigma_{EXEMPT-\mathcal{K}}$, and $\Sigma_{EXEMPT-\mathcal{H}}$ are sets of all strands with the respective traces $Tr_{EXEMPT-\mathcal{S}}$, $Tr_{EXEMPT-\mathcal{K}}$, and $Tr_{EXEMPT-\mathcal{H}}$.
\end{definition}

\subsubsection{Security Goals} 
\label{ssec:E-goals}
Because the exemption protocol has the same overall objective as the basic query protocol, we modify the security goals in Appendix~\ref{ssec:Q-goals} to incorporate traces and additional assumptions specific to the exemption protocol.
Key additions include agreement on the exemption sequence $M(s_E)$ and the authentication code $AuthCode$.
As we do in the query protocol, we assert the confidentiality of the client sequence $M(s)$.

\begin{definition}[Exemption SG2: Synthesizer Perspective]
\label{def:SG2-exempt-synthclient}
Let \textbf{exempt-orig} be the conjunction of \textbf{orig} in Definition~\ref{def:SG2-synthclient} and the following origination assumptions on $\vec{E}$:
\begin{enumerate}
    \item $Non(sk(\mathcal{B}))$, $Non(sk(\mathcal{C}))$
    \item $Uniq(S_E)$
    \item $PenNon(AuthCode)$, where $PenNon$ indicates {penetrator non-origination}.
\end{enumerate}
\textit{Penetrator non-originating} terms such as authentication codes or passwords are unknown and unguessable to the adversary, but known to all regular strands. 

An exemption protocol strand space $\mathcal{P}$ satisfies SG2 from the synthesizer perspective iff $Conf_\mathcal{P}(Tr_{EXEMPT-\mathcal{S}}(\vec{E}), M(s), \mathbf{exempt-orig}$) is true.
\end{definition}

\begin{definition}[Exemption SG2: Keyserver Perspective]
\label{def:SG2-exempt-keyserver}
Let \textbf{exempt-orig} be the conjunction of \textbf{orig} in Definition~\ref{def:SG2-keyserver} and the following origination assumptions on $\vec{E}$:
\begin{enumerate}
    \item $Non(sk(\mathcal{B}))$, $Non(sk(\mathcal{C}))$
    \item $Uniq(S_E)$
    \item $PenNon(AuthCode)$, where $PenNon$ indicates penetrator non-origination.
\end{enumerate}

An exemption protocol strand space $\mathcal{P}$ satisfies SG2 from the keyserver perspective iff $Conf_\mathcal{P}(Tr_{EXEMPT-\mathcal{K}}(\vec{E}), M(s), \mathbf{exempt-orig}$) is true.
\end{definition}

\begin{definition}[Exemption SG2: Keyed-Hash Database Perspective]
\label{def:SG2-exempt-database}
Let \textbf{exempt-orig} be the conjunction of \textbf{orig} in Definition~\ref{def:SG2-database} and the following origination assumptions on $\vec{E}$:
\begin{enumerate}
    \item $Non(sk(\mathcal{B}))$, $Non(sk(\mathcal{C}))$
    \item $Uniq(S_E)$
    \item $PenNon(AuthCode)$, where $PenNon$ indicates penetrator non-origination.
\end{enumerate}

A query protocol strand space $\mathcal{P}$ satisfies SG2 from the keyserver perspective iff $Conf_\mathcal{P}(Tr_{EXEMPT-\mathcal{H}}(\vec{E}), M(s), \mathbf{exempt-orig}$) is true.
\end{definition}

For the exemption protocol, we extend the agreement criteria of Definition~\ref{def:SG3-synthclient} and Definition~\ref{def:SG3-database} to include the exempt sequence $M(s^E)$ and the authentication code.
Because the exemption protocol also builds on the flawed SCEP, there are similar issues when attempting to prove agreement from {\myH}'s perspective.

\begin{definition}[Exemption SG3: Synthclient Perspective]
\label{def:SG3-exempt-synthclient}
Let \textbf{exempt-orig} be the conjunction of \textbf{orig} in Definition~\ref{def:SG3-synthclient} and the following origination assumptions on $\vec{E}$:
\begin{enumerate}
    \item $Non(sk(\mathcal{B}))$, $Non(sk(\mathcal{C}))$
    \item $Uniq(S_E)$
    \item $PenNon(AuthCode)$, where $PenNon$ indicates penetrator non-origination.
\end{enumerate}
An exemption protocol strand space $\mathcal{P}$ satisfies SG3 from the synthesizer perspective iff $Agree_\mathcal{P}(Tr_{EXEMPT-\mathcal{S}}(\vec{E}), Tr_{EXEMPT-\mathcal{H}}(\vec{E})$,\linebreak $[M(s)^k,M(s^E)^k, AuthCode, Resp], \mathbf{exempt-orig})$ is true.
\end{definition}

\begin{definition}[Exemption SG3: Keyed-Hash Database Perspective]
\label{def:SG3-exempt-database}
Let \textbf{exempt-orig} be the conjunction of \textbf{orig} in Definition~\ref{def:SG3-database} and the following origination assumptions on $\vec{E}$:
\begin{enumerate}
    \item $Non(sk(\mathcal{B}))$, $Non(sk(\mathcal{C}))$
    \item $Uniq(S_E)$
    \item $PenNon(AuthCode)$, where $PenNon$ indicates penetrator non-origination.
\end{enumerate}
An exemption protocol strand space $\mathcal{P}$ satisfies SG3 from the database perspective iff $Agree_\mathcal{P}(Tr_{EXEMPT-\mathcal{H}}(\vec{E}), Tr_{EXEMPT-\mathcal{S}}(\vec{E})$, \linebreak $[M(s)^k,M(s^E)^k, AuthCode, Resp], \mathbf{exempt-orig})$ is true.
\end{definition}

\subsubsection{Analysis} 
\label{ssec:E-analysis}

As we do for the basic query protocol (Appendix~\ref{ssec:Q-analysis}), we now state and prove theorems that assert our exemption protocol confidentiality and agreement goals.

\begin{theorem}
[Confidentiality of $M(s)$]
\label{thrm:exempt-conf}
If $\mathcal{P}$ is any exemption protocol strand space,
then $\mathcal{P}$ satisfies SG2 from each role's perspective (Definition~\ref{def:SG2-exempt-synthclient}, Definition~\ref{def:SG2-exempt-keyserver},
Definition~\ref{def:SG2-exempt-database}).
\end{theorem}

\begin{proof}[Proof by Enumeration]
Within each perspective, there exist no unique execution models that contradict the confidentiality of the client sequence $M(s)$.
\end{proof}

\begin{theorem}
[Agreement on $M(s)^k$, $M(s^E)$, $AuthCode$, $Resp$]
\label{thrm:exemp-agree}
If $\mathcal{P}$ is any exemption protocol strand space,
then $\mathcal{P}$ does \textbf{not} satisfy SG3 from the keyed-hash database perspectives
(Definition~\ref{def:SG3-exempt-database}).
\end{theorem}

\begin{proof}[Proof by Counterexample]
There exist multiple shapes in the perspective of {\myH} that contradict agreement on $M(s)^k, M(s^E), AuthCode, Resp$ between a strand of role {\myH} and another legitimate strand.
\end{proof}


\clearpage
\subsection{Other Security Issues}
\label{sec:other}

This section briefly identifies several additional security design issues 
of the SecureDNA system we reviewed
that are noteworthy of consideration, even if they do not
present imminent high risk.
Some of these issues highlight tradeoffs inherent in engineering systems,
including tradeoffs between security strength and user requirements.

(1)~The published descriptions~\cite{SecureDNA2024} and source code of the SecureDNA system
do not describe how they would replace the shared DOPRF key $k$ should it ever be compromised.  All security systems should address re-keying procedures.
The SecureDNA team explained that, should they need to replace $k$,
they would do what they do whenever they regenerate {\myH}:
generate a new $k$ and regenerate {\myH}, which they say they can do quickly.


(2)~In the current system implementation, the DOPRF key is generated \textit{centrally}, not in a \textit{distributed} fashion as per their design, creating a single point of failure. The SecureDNA team stated that a future update will address this issue.

(3)~Some intended potentially important security features were not yet implemented
at the time of our review.
These features include \textit{proactive} secret sharing~\cite{herzberg1995proactive} 
(in which key shares are rotated to mitigate the threat that an adversary might eventually compromise all keyservers, a few at a time).
The SecureDNA team explained that, instead of rotating the key shares, 
they can re-key and regenerate {\myH}.
They also explained that a new version of the code now checks for high rate limits, and
that they monitor and enforce rate limits 
at servers when clients make requests. 


(4)~The system encrypts stored keys with user-entered passphrases, but writes
these passphrases as plain text to a provider-specified location.
This location might be an ordinary file,
possibly defeating any benefit from encrypting the keys.
While DNA synthesis providers have the option to institute stronger security practices (e.g., involving a password manager or TPM), these practices are not mandated
and SecureDNA cannot control them.

(5)~Some of the auditing and record-keeping mechanisms 
depend in part on the integrity and assured operation of email. 
For example, the exemption-handling protocol 
sends an alert to {\myB} via email whenever it processes an ELT [E24].
A powerful attacker might be able to block email notifications.
The SecureDNA team explained that providers and {\myB} typically use cloud-based email
services and that all of the auditing emails are encrypted.

(6)~Rather than using a well-vetted library to support certificate infrastructure,
SecureDNA creates its own custom infrastructure, unnecessarily increasing complexity
and risk of security errors.  Given that the X.509 standard~\cite{X509} includes
an optional user-defined field, we feel there is no compelling reason to create
a custom infrastructure (see also Section~\ref{sec:discussion}).


(7)~In the basic-request protocol, {\myS}'s response to {\myC} includes
information that is potentially dangerous, especially for novel hazards: 
if the request is denied, {\myS} includes
in its response which sequence(s) in the requested list of sequences are hazardous. 
It also states the pathogen from which the hazardous sequence came
and why that pathogen is dangerous.
The SecureDNA team explained that clients and providers demand this information,
and novel hazards are treated differently.


(8)~The implemented DOPRF is not post-quantum secure~\cite{bernstein2017post}.
The SecureDNA team doubted that any post-quantum secure DOPRF would meet their
time and space performance requirements.  They also said that a post-quantum adversary
would be able to break TLS.

\clearpage
\onecolumn
\clearpage 
\subsection{Summary Results from Formal-Methods Analyses of SecureDNA Protocols}
\label{append:results}

\begin{table}[h]
\caption{Security Properties of SCEP and SCEP+. 
For each model, a check ({\checkmark}) indicates that, for all possible execution models under corresponding goal's security assumptions, the logical security goal holds. 
A crossmark ({\myno}) indicates that CPSA finds a counterexample that disproves the corresponding security goal.
See Definition~\ref{def:scep-sw-secure} and Definition~\ref{def:scep-ws-secure} for the security goal definitions.}
\label{tab:results1}
\begin{center}
\begin{tabular}{l|l|l}
Model & {\myS}-{\myW} Secure & {\myW}-{\myS} Secure\\
\hline
SCEP  & {\checkmark} & {\myno}\\
SCEP+  & {\checkmark} & {\checkmark}\\
\end{tabular}
\end{center}
\vspace*{-16pt}
\end{table}

\bigskip \bigskip

\begin{table}[h]
\caption{Security Properties of Basic Query and Exemption Query. To achieve confidentiality of $M(s)$ and agreement on $M(s)^k$, $M(s^E)$, $AuthCode$, $Resp$, the protocol must satisfy the corresponding security goal definitions in Appendix~\ref{ssec:Q-goals} (Basic Query protocol) and Appendix~\ref{ssec:E-goals} (Exemption Query protocol).}
\label{tab:results2}
\begin{center}
\begin{tabular}{l|p{0.85in}|p{0.85in}}
Model & Confidentiality of $M(s)$ & Agreement on $M(s)^k$, $M(s^E)$, $AuthCode$, $Resp$\\
\hline
Basic Query (No TLS) & {\checkmark} & {\myno}\\
Basic Query  & {\checkmark} & {\myno}\\
Exemption Query & {\checkmark} & {\myno}\\
Basic Query (SCEP+) & {\checkmark} &{\checkmark}
\end{tabular}
\end{center}
\vspace*{-16pt}
\end{table}

\vfill

\subsection{Rate-Limiting Attack by Corrupt Keyserver}
\label{append:rate-attack}

We provide more details about the rate-limiting attack described in
Section~\ref{ssec:R-attacks}.
Whenever a legitimate synthesizer {\myS} makes a request to a corrupt keyserver {\myKp}, 
the following rate-limiting attack presents itself to {\myKp}.  
A similar attack is possible if {\myS} connects with
a corrupt {\myH}.

\bigskip
\begin{tabular}{lll}
[1] & {\myS} $\rightarrow$  {\myKp}: & Nonce({\myS}), Token({\myS})
\end{tabular}
\bigskip

Following this exchange, {\myKp} now knows Nonce({\myS}) and possesses Token({\myS}). 
Assume that {\myS} is far below their current rate limit. 
{\myK} now has the opportunity to masquerade as {\myS} to other keyservers 
or to {\myH}. 
Let {\myW} denote any legitimate responding infrastructure (e.g., {\myH}).



\bigskip
\begin{tabular}{lll}
[2] & {\myKp} $\rightarrow$  {\myW}: & Nonce({\myS}), Token({\myS})\cr
[3] & {\myKp} $\leftarrow$  {\myW}: & 
Cookie, Nonce({\myW}), Token({\myW}), $x =$ Sign[Hash(``server-mutauth'', Nonce({\myS}), Nonce({\myW}), \cr
& & Token({\myW})), Privk({\myW})]
\end{tabular}
\bigskip

{\myKp} also knows Nonce({\myW}) and Token({\myW}), and has a signed copy of the 
hash that {\myS} expects. 
{\myKp} can now exploit the weakness of SCEP to trick {\myS} into authenticating {\myKp} to {\myW}.

\bigskip
\begin{tabular}{lll}
[4] & {\myS} $\leftarrow$ {\myKp}: & Cookie, Nonce({\myW}), Token({\myW}), $x$\cr
[5] & {\myS} $\rightarrow$  {\myKp}: & Cookie, $y =$ 
Sign[Hash(``client-mutauth'', Nonce({\myS}), Nonce({\myW}), Token(S)), Privk({\myS})]
\end{tabular}
\bigskip

Step~5 is important: {\myKp} obtains the critical value $y$.
Knowledge of $y$ enables {\myKp} to masquerade as {\myS} to {\myW}.

\bigskip
\begin{tabular}{lll}
[6] & {\myKp} $\rightarrow$  {\myW}: & Cookie, $y$
\end{tabular}
\bigskip

{\myW} now believes {\myKp} is {\myS}. 
Thus, {\myKp} is able to use the rate limit budget of {\myS} for their own requests to other keyservers or {\myH}. 
Behold {\myKp} can perform this attack on many different instances of {\myS} simultaneously.
This attack is within SecureDNA's adversarial model, but Baum et al.\ \cite{baum2024CIC,baum2020cryptographic}
did not consider rate-limiting attacks.


\clearpage
\subsection{Ladder Diagrams of SecureDNA Protocols}
\label{append:ladder}

\begin{figure}[h]
\includegraphics[width=\textwidth]{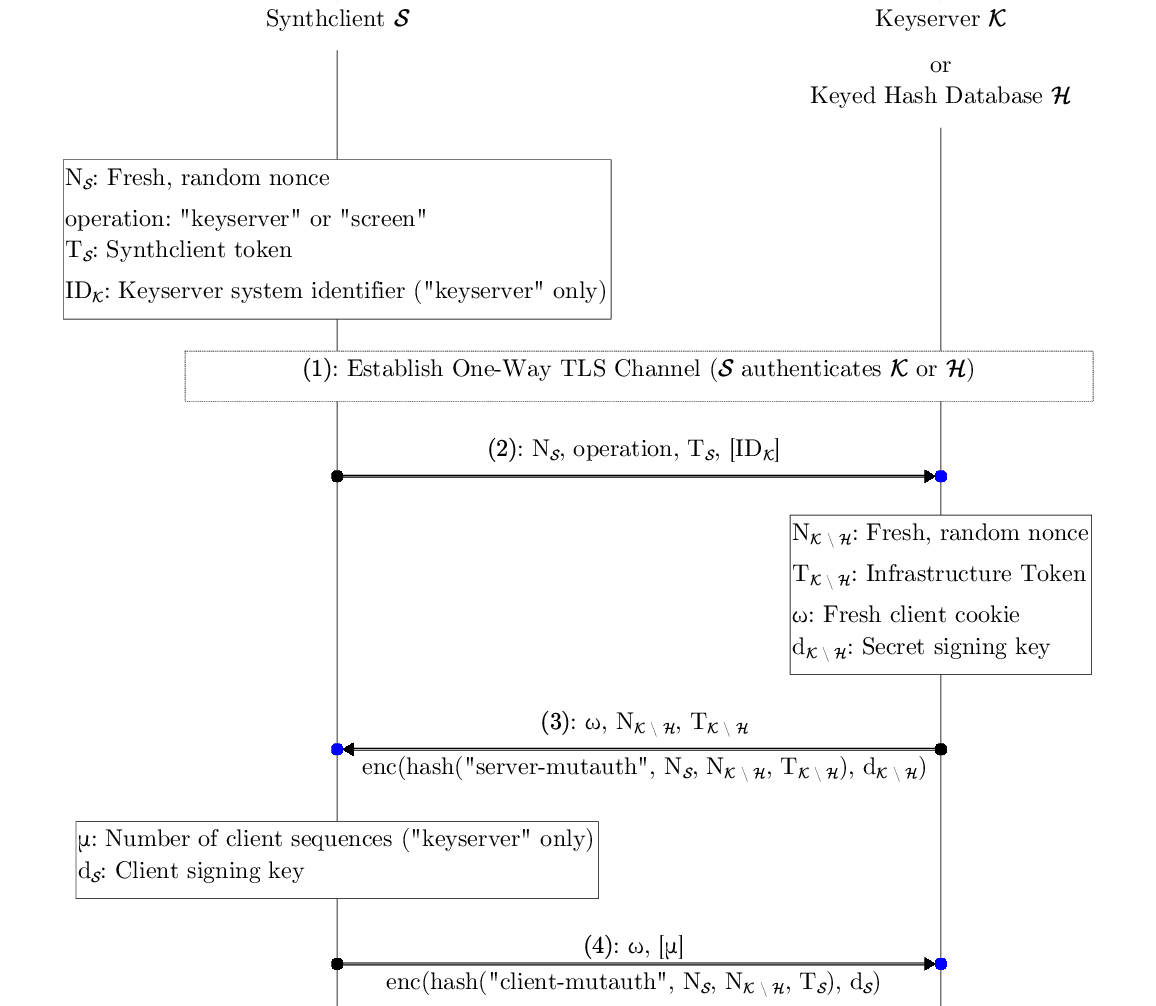}
\caption{Ladder diagram of the custom mutual authentication \textit{(SCEP)} protocol. Vertical lines correspond to communicating roles ({\myS, \myK, \myH}). Arrows between vertical lines indicate roles transmitting and receiving protocol messages. SCEP assumes an existing one-way authenticated TLS channel between the synthesizer {\myS}, and the keyserver {\myK} or the keyed hash database {\myH}. Some message components appear only when {\myS} communicates with {\myK}; we indicate these components by surrounding them with brackets ({[]}). The responder nonce $N$, token $T$, and secret signing key $d$ correspond to the entity with which {\myS} communicates. Upon completing SCEP, {\myS} obtains a request cookie {\mycookie}, which {\myS} transmits in a subsequent request to {\myK} or {\myH}.}
\label{fig:auth-ladder}
\end{figure}


\clearpage

\vfill

\begin{figure}[b]
\includegraphics[width=\textwidth]{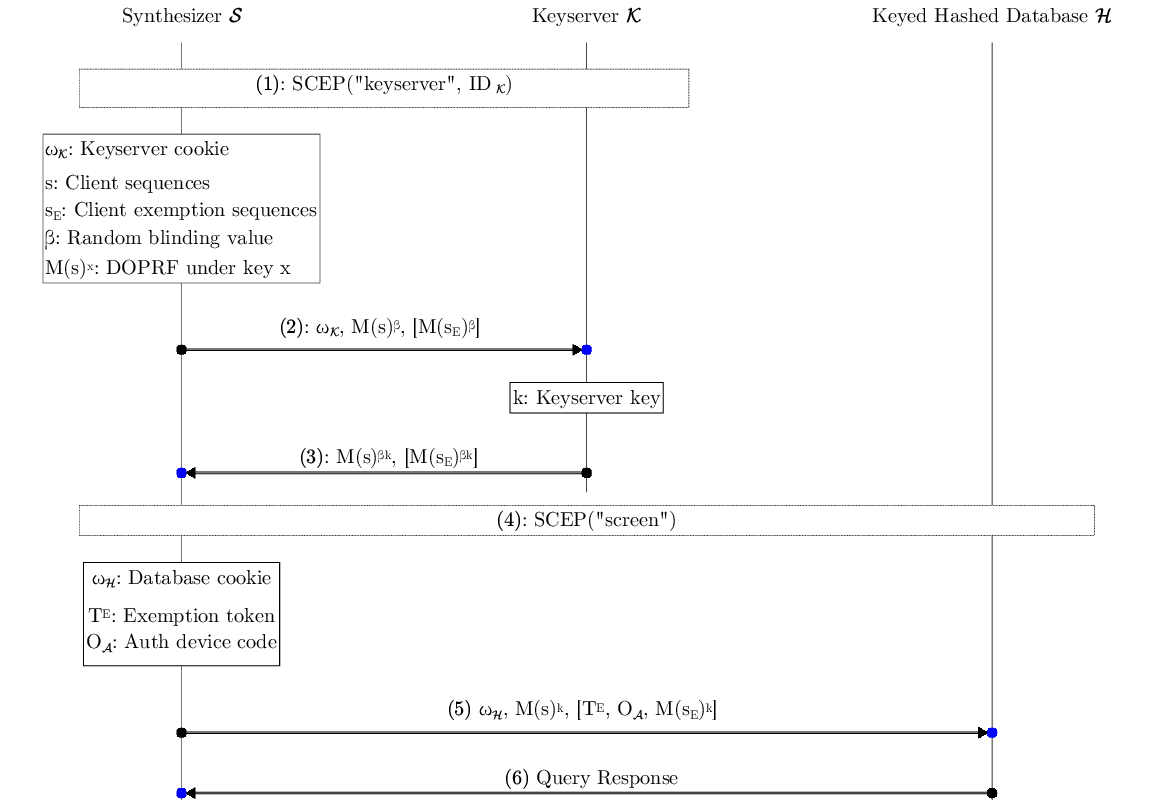}
\caption{
Ladder diagrams of the basic order-request and exemption-handling protocols. Because {\myS} communicates with {\myK} and {\myH} over separate, one-way authenticated TLS connections, {\myS} completes the SCEP mutual authentication protocol (see Figure~\ref{fig:auth-ladder}) in Steps~1 and 4 to obtain request cookies $\omega_\mathcal{K}$ and $\omega_\mathcal{H}$. When completing SCEP with {\myK}, {\myS} includes {\myK}'s identifier $I\!D_\mathcal{K}$. In Step~2, {\myS} transmits to {\myK} the cookie $\omega_\mathcal{K}$, blinded sequences $M(s)^\beta$ and (optionally) the blinded exempt sequences $M(s_E)^\beta$. In Step~5, {\myK} transmits $\omega_\mathcal{H}$, $M(s)^k$, and an optional tuple (exemption token $T^E$, second-factor authentication code $O_A$, and exemption sequence $M(s_E)^k$), to which {\myK} receives a query response that grants or denies the synthesis request in Step~6.
}
\label{fig:two-ladders}
\end{figure}

\clearpage
\subsection{CPSA Source Code Snippets}
\label{append:snippets}

\begin{figure}[h] 
\begin{Verbatim}
  (defrole synthesizer
    (vars
      (x expt) (y rndx) (cr sr text) 
      (r-s r-s2 r-k r-d t random256)
      (c-id k-id m-root i-root name)
      (seq blind k rndx)
      (b-s b-k bundle-data)
      (resp query-response)
    )
    (trace
      (Synthesizer2Server send recv x y cr sr
        c-id k-id m-root i-root
        r-s r-k t
        b-s b-k k)
      )
    (uniq-gen y)
    )
\end{Verbatim}
\caption{
CPSA specification for the \textit{synthesizer} role in the SCEP protocol.
This role declares variables and defines trace events 
(see Figure~\ref{fig:cpsa_src_clientM} for the definition of the 
\texttt{Synthesizer2Server} macro) 
that carry out the SCEP protocol.
Variable definitions comprise s-expressions that list labels followed by a sort (type).
To model the responding \textit{infrastructure server} role, we reverse the direction of each trace event by reversing the order of the \texttt{send} and \texttt{recv} inputs to the \texttt{Synthesizer2Server} macro.
}
\label{fig:cpsa_src_client}
\end{figure}

\bigskip \bigskip

\begin{figure}[h] 
\begin{Verbatim}
(defmacro (Synthesizer2Server 
     d1 d2 x y cr sr
     client-id server-id manufacturer-root infrastructure-root
     client-nonce server-nonce cookie
     client-bundle server-bundle key)
  (^
   (TLS_EDH d1 d2 x y cr sr server-id (pubk server-id) infrastructure-root)
   (d1 (enc (cat client-nonce "keyserver" server-id
      (Token client-id manufacturer-root client-bundle))
      (ClientWriteKey (MasterSecret (exp (gen) (mul x y)) cr sr))))
   (d2 (enc (ServerMutauthReq
      server-id client-id manufacturer-root infrastructure-root
      server-nonce client-nonce cookie
      client-bundle server-bundle)
      (ServerWriteKey (MasterSecret (exp (gen) (mul x y)) cr sr))))
   (d1 (enc (ClientMutauthResp
      server-id client-id manufacturer-root infrastructure-root
      server-nonce client-nonce cookie
      client-bundle server-bundle)
     (ClientWriteKey (MasterSecret (exp (gen) (mul x y)) cr sr))))
))
\end{Verbatim}
\caption{
CPSA macro that defines a trace for the SCEP protocol. This macro depends on other macro definitions (\texttt{TLS\_EDH, Token, ClientWriteKey, MasterSecret, ServerWriteKey, ServerMutauthReq, ClientMutauthResp}) available in our GitHub repository~\cite{anonGitHub}.
When executing, CPSA expands these macros to build the full trace for an SCEP synthesizer or infrastructure server strand.
}
\label{fig:cpsa_src_clientM}
\end{figure}

\vfill

\clearpage
\subsection{CPSA Shapes}
\label{append:shapes}

\begin {figure}[h]
\includegraphics[width=\textwidth]{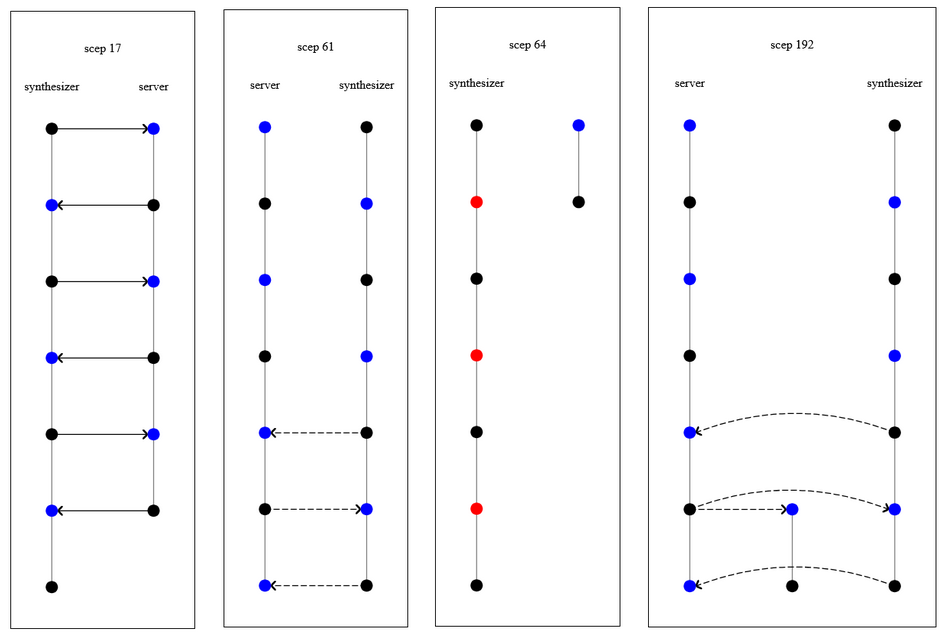}
\caption{Selected (non-exhaustive) CPSA shapes for SCEP executions. 
Each shape corresponds to a unique execution model that CPSA discovers during execution.
Shape 17 provides evidence that, from the perspective of a legitimate synthesizer, there exists an exchange of messages with a legitimate infrastructure server such that both parties agree on each value.
Shape 61 illustrates the weakness of SCEP: the dashed lines indicate that penetrator strands are manipulating messages in transit, proving that the synthesizer and the server do not agree on the session parameters.
Shape 64 is a ``dead'' shape in which a listener strand listens for the session cookie $\omega$. 
Because CPSA cannot satisfy (or ``realize'') the nodes for the synthesizer strand, we prove that, within the execution model, the adversary does not learn $\omega$.
Shape 192 illustrates that, due to the weakness of the SCEP (see Shape 61), 
the value $\omega$ leaks to a listener strand.
To prove our theorems, we apply logical formulae on each such shape resulting from executing CPSA to test a security goal.}
\label{fig:code}
\end{figure}

\vfill

\clearpage
\subsection{Endnotes About SecureDNA Source Code }
\label{append:endnotes}
\label{append:notes}

\medskip
These endnotes provide support for our descriptions of the SecureDNA system and its protocols.
The files are located in the crates folder of the SecureDNA source code~\cite{SecureDNAweb}.


\noindent \begin{tabular}{llll} \\
\textbf{Note} & \textbf{File} & \textbf{Line(s)} & \textbf{Comment} \\ \hline
1 & certificates/src/certificates/inner/common.rs & 25 & Cert subject fields\\
2 & certificates/src/certificates/inner/common.rs & 178 & Cert issuer fields\\
3 & certificates/src/certificates/inner/common.rs & 92 & Exemption cert subject fields\\
4 & keyserver/src/server.rs & 172 & Keyserver {\myK} main loop\\
5 & certificates/src/tokens/infrastructure/keyserver.rs & 51 & Keyserver token subject fields\\
6 & certificates/src/tokens/infrastructure/keyserver.rs & 144 & Keyserver token issuer fields\\
7 & hdbserver/src/server.rs & 184 & Keyed Hash Database {\myH} main loop\\
8 & certificates/src/tokens/infrastructure/database.rs & 46 & Database token subject fields\\
9 & certificates/src/tokens/infrastructure/database.rs & 137 & Database token issuer fields\\
10 & synthclient/src/shims/server.rs & 155 & Synthclient {\myS} main loop\\
11 & certificates/src/tokens/manufacturer/synthesizer.rs & 83 & Synthesizer token subject fields\\
12 & certificates/src/tokens/manufacturer/synthesizer.rs & 205 & Synthesizer token issuer fields\\
13 & certificates/src/tokens/exemption/et.rs & 58 & Exemption token subject fields\\
14 & certificates/src/tokens/exemption/et.rs & 218 & Exemption token issuer fields\\
15 & scep/src/steps.rs & 142-162 & {\myK} or {\myH} side SCEP part 1\\
16 & scep/src/steps.rs & 231-270 & {\myS} side SCEP\\
17 & scep/src/steps.rs & 368-383 & {\myK} or {\myH} side SCEP part 2\\
18 & certificates/src/chain\_item.rs & 130 & Cert chain validation part 1\\
19 & certificates/src/chain\_item.rs & 74 & Cert chain validation part 2\\
20 & scep/src/steps.rs & 410-436 & Rate limit check\\
21 & certificate\_client/src/shims/create\_token.rs& 61-77 & Synthesizer token create options\\
22 & scep/src/mutual\_authentication.rs & 1-82 & Various mutual authentifation functionality\\
23 & minhttp/src/mpserver/common.rs & 50 &Default rustls configuration has 0-rtt and\\ &&& session resumption disabled\\
24 & certificates/src/token/exemption/et.rs & 299 & Emails to notify function
\end{tabular}


\clearpage
\subsection{Acronyms and Abbreviations}
\label{append:abbrevs}

\noindent \begin{tabular}{ll} \\[-6pt] 
CA & certificate authority\\
CPSA & Cryptographic Protocol Shapes Analyzer\\
CRISPR & Clustered Regularly Interspaced Short Palindromic Repeats\\
DH & Diffie-Hellman\\
DOPRF & distributed oblivious pseudorandom function\\
DoS & denial of service\\
DY & Dolev-Yao\\
ELT & exemption list token\\
INSuRE & Information Security Research Education\\
MAC & message authentication code\\
MitM & man-in-the-middle\\
mTLS & TLS with mutual authentication\\
NS & Needham-Schroeder\\
PCAP & packet capture\\
PKI & public-key infrastructure\\
RSA & Rivest, Shamir, Adleman\\
SCEP & SecureDNA's Server Connection Establishment Protocol\\
SCEP+ & our improved version of SCEP\\
SDF & SecureDNA Foundation\\
SDS & SecureDNA system\\
SOC & Security Operations Center\\
TLS  & Transport Layer Security\\
TPM & Trusted Platform Module\\
UC & universal composability\\ 
UMBC & University of Maryland, Baltimore County\\ 
0-RTT & zero roundtrip time\\
\end{tabular}

\clearpage 
\setcounter{tocdepth}{3}
\tableofcontents 

\end{document}